\documentclass{article}

\usepackage{arxiv}

\usepackage[utf8]{inputenc}
\usepackage{lmodern}
\usepackage[svgnames]{xcolor}
\usepackage{graphicx}
\graphicspath{{images/}}
\usepackage{comment}
\usepackage{tikz}
\usetikzlibrary{math, shapes}
\usepackage{pgfplots}
\pgfplotsset{compat=1.9,legend style ={font=\footnotesize}, width=7cm, axis lines = left, yminorticks=false}
\usepackage{epstopdf}
\usepackage{subcaption}
\usepackage{extarrows}
\usepackage{amssymb,amsmath,amsfonts,amsthm,bm,mathtools,cuted,bbold}
\usepackage{tabularx}
\usepackage{booktabs}
\usepackage[ruled,vlined,linesnumbered]{algorithm2e}

\newcommand{\T}{^{\rm {T}}}
\newcommand{\E}[1]{\mathbb{E}\left\{ #1 \right\}}
\newcommand{\mc}[1]{\mathcal{#1}}
\newcommand{\mb}[1]{\mathbf{#1}}
\newcommand{\code}[1]{\texttt{#1}}

\DeclareMathOperator*{\argmax}{arg\,max}
\DeclareMathOperator*{\argmin}{arg\,min}

\newcommand{\FS}[1]{\textcolor{black}{ {#1}}}

\newcommand{\mm}[1]{\textcolor{black}{ {#1}}}
\renewcommand{\qed}{\hfill \ensuremath{\Box}}

\newtheorem{proposition}{Proposition}

\title{Power Minimization of Downlink Spectrum Slicing for eMBB and URLLC Users} 

\author{Fabio Saggese \\
\textit{Dept. of Electronic System} \\
\textit{Aalborg University}\\
Aalborg, Denmark \\
\texttt{fasa@es.aau.dk}
\And
Marco Moretti \\
\textit{Dept. of Information Engineering} \\
\textit{University of Pisa}\\
Pisa, Italy \\
\code{marco.moretti@unipi.it}
\And
Petar Popovski \\
\textit{Dept. of Electronic Systems} \\
\textit{Aalborg University}\\
Aalborg, Denmark \\
\code{petarp@es.aau.dk}
}

\date{}

\begin{document}
\maketitle

\begin{abstract}
    %5G technology allows heterogeneous services to share the wireless spectrum within the same radio access network. In this context, \emph{spectrum slicing} of the shared radio resources is a critical task to guarantee the performance of each service. We analyze a downlink communication serving two types of traffic: enhanced mobile broadband (eMBB) and ultra-reliable low-latency communication (URLLC). Due to the nature of low-latency traffic, the base station knows the channel state information (CSI) of the eMBB users, while having statistical CSI for the URLLC users. We study the power minimization problem employing orthogonal multiple access (OMA) and non-orthogonal multiple access (NOMA) schemes. Based on this analysis, we propose two algorithms: a lookup table-based and a block coordinated descent (BCD). We show that the BCD is optimal for the URLLC power allocation. The numerical results show that NOMA leads to a lower power consumption compared to OMA, except when the average channel gain of the the URLLC user is very high. For the latter case, the optimal approach depends on the channel condition of the eMBB user. Even when OMA attains the best performance, the gap with NOMA is negligible. This shows the capability of NOMA to reduce the power consumption in practically every condition.
    5G technology allows heterogeneous services to share the wireless spectrum within the same radio access network. In this context, \emph{spectrum slicing} of the shared radio resources is a critical task to guarantee the performance of each service. We analyze a downlink communication serving two types of traffic: enhanced mobile broadband (eMBB) and ultra-reliable low-latency communication (URLLC). Due to the nature of low-latency traffic, the base station knows the channel state information (CSI) of the eMBB users while having statistical CSI for the URLLC users. We study the power minimization problem employing orthogonal multiple access (OMA) and non-orthogonal multiple access (NOMA) schemes. Based on this analysis, we propose a lookup table-based approach and a block coordinated descent (BCD) algorithm. We show that the BCD is optimal for the URLLC power allocation. The numerical results show that NOMA leads to lower power consumption than OMA, except when the average channel gain of the URLLC user is very high. For the latter case, the optimal approach depends on the channel condition of the eMBB user. Even when OMA attains the best performance, the gap with NOMA is negligible, showing the capability of NOMA to reduce power consumption in practically every condition.
\end{abstract}
% \begin{IEEEkeywords}
% NOMA, RAN slicing, eMBB, URLLC, Power saving
% \end{IEEEkeywords}

\section{Introduction}
\label{sec:intro}
The plethora of new services promised by 5G and beyond systems calls for the coexistence of very heterogeneous types of traffic on the same physical network.
Some services pose strict requirements in terms of latency, others in terms of high reliability or of the huge number of devices connected to the network, while the majority of \emph{traditional} services still require high bandwidth and data rate.
To address the complexity of such a vast space of requirements, 5G standardized three generic service types: enhanced mobile broadband (eMBB), massive machine-type communications (mMTC), and ultra-reliable low-latency communications (URLLC)~\cite{3gpp:access}. Note that a realistic service could require any combination of the aforementioned generic service types.
At the network level, \emph{network slicing} deals with the partitioning of the physical network infrastructure into different end-to-end isolated virtual networks able to support specific service requirements for the various use cases~\cite{Zhang2017}. Similarly, at the physical layer \emph{spectrum slicing} deals with seamlessly allocating the radio spectrum to serve users with heterogeneous requirements and is a fundamental Radio Access Network (RAN) task.

% \FS{The authors of~\cite{Ksentini2017} proposed the idea of enforcing the slices through the RAN resource allocation onto virtualized resource block, representing a set of physical resource blocks (RB).}
The problem of physical resource allocation for RAN slicing has been addressed in~\cite{Doro2019}, where the authors allocate resource blocks (RB) to different base stations (BS) to meet the demand of mobile network operators. This work gave one of the first formalizations of the RAN slicing problem, even if it did not take into account any specific physical layer requirements. 
Henceforth, several works treated the problem of resource allocation algorithms to multiplex eMBB and URLLC services. In~\cite{Anand2020}, the joint resource allocation problem for eMBB-URLLC slicing is addressed by employing different puncturing models, where the reliability of the URLLC transmission is always considered met. The authors of~\cite{Elsayed2019} and~\cite{Alsenwi2021} propose two deep reinforcement learning techniques to allocate resources to the two services, employing orthogonal resources and preemption/puncturing, respectively. Perfect knowledge of CSI is a crucial prerequisite for these algorithms to work properly. 
\FS{The authors of~\cite{Huang2022deluxe} propose a machine learning solution based on link adaptation to minimize the impact of URLLC traffic on the eMBB transmission.}
All these works adopt the current OMA 5G standard for dynamic resource sharing between eMBB and URLLC, either by the use of puncturing or by employing non-overlapping time/frequency resources~\cite{3gpp:access}.
However, NOMA has proven to outperform the OMA in many applications~\cite{Xu2015} and can also be applied to spectrum slicing, as introduced in the uplink communication framework presented in~\cite{Popovski2018}. 

In general, most recent works on spectrum slicing for NOMA communications, such as \cite{Kalor2019}\nocite{Li2019}-\cite{Chiarotti21},  focus on the \emph{uplink} direction because of the simpler implementation of the successive interference cancellation (SIC) strategy.
The work in~\cite{Kalor2019} compares the performance of NOMA and OMA for URLLC devices with different latency requirements and limited feedback by the receiver.
In~\cite{Li2019}, a reinforcement learning algorithm decides whenever to use OMA or NOMA for dynamic multiplexing of eMBB-URLLC data streams, allocating the transmission power on the information of the channel state. 
The authors in~\cite{Chiarotti21} investigate the use of OMA and NOMA for intermittent and broadband services. The reliability is enforced by employing packet coding on the binary erasure channel, while the latency is modeled, taking into account reliability-latency and age of information metrics.

\subsection{Contributions}
In this paper, we analyze the problem of slicing the spectrum for eMBB and URLLC traffics in the \emph{downlink} direction of a NOMA system and compare the performance with a  more traditional approach based o orthogonal multiplexing.
\mm{Our research addresses a simplified scenario with two users only; considering the inherent complexity of the problem, the goal of this paper is to outline the most important relationships between the two types of traffic. In an effort of abstraction, this setting, although simple, preserves all the most significant factors that affect the problem of coexistence of the eMBB and URLLC traffics in the downlink of a NOMA system.} 

\mm{The eMBB service aims to maximize its throughput without any latency requirements; on the contrary, the URLLC service demands mission-critical, reliable communication with a hard latency constraint. 
Considering the intermittent nature and the stringent latency constraints, the estimation of the instantaneous CSI is infeasible for URLLC packets, and we assume to possess only the knowledge of statistical CSI. Accordingly, the performance of the URLLC user is expressed in terms of its \emph{outage probability}, which depends on the number of resources selected, the spectral efficiency of the transmission, and the mean signal-to-noise ratio at the receiver. On the other hand, the instantaneous CSI for the eMBB user is available, and its performance is expressed in terms of achievable rate.}
%Results for algorithms optimized for both  NOMA and OMA are compared. In particular, we focus on OMA and NOMA allocation, remarking that the OMA approach can be extended to encompass puncturing. In fact, in our model, the specific case of puncturing is absorbed by the OMA model, being inherently a form of orthogonal transmission.
%Thus, using puncturing can be modeled as the OMA transmission depicted in the remainder of the paper. % Moreover, to guarantee the performance required by the eMBB transmission, some protection mechanisms on the punctured resources must be performed, e.g., erasure~\cite{Popovski2018}, or packet-level codes~\cite{Chiarotti21}. These protection methods employ more power than strictly needed to add redundancy to the punctured data stream. However, this is not an optimal approach for our setting, considering that we are interested in saving power while guaranteeing the traffic requirements. 
% our approach 

\mm{With respect to the existing literature on the subject, employing a NOMA approach to the problem of downlink RAN slicing may offer significant gains but also presents many challenges. In the first place, assuming heterogeneous requirements among users dictates a new approach to SIC for NOMA. Rather than following an optimal cancellation order based on channel quality, we are constrained by the strict latency requirements of the URLLC traffic to always apply SIC at the eMBB receiver.
Secondly, the propagation conditions for URLLC transmissions are so adverse due to the presence of eMBB interference and the lacking knowledge of instantaneous CSI that NOMA can be reasonably applied only in a multi-carrier approach that exploits the frequency diversity of independent parallel channels.}
\mm{As a consequence, the derivation of the outage probability for the URLLC traffic, similar to the problem addressed in~\cite{Wei2017mcnoma} for a NOMA system with imperfect CSI,  is a difficult task that can not be solved by applying the existing literature either in closed~\cite{Yilmaz2010} or approximated form~\cite{Bai2013, Coon2015} and requires a novel approach.}

\mm{As in most scenarios where interference plays a major role, the total consumed power is a critical indicator to evaluate the performance of the system. Accordingly, we formulate the problem of eMBB-URLLC spectrum slicing with the objective of minimizing the overall power under the conflicting constraints dictated by the two types of traffic, i.e.,  reliable URLLC transmissions and large spectral efficiency for eMBB traffic.
Numerical results investigate various aspects of the interaction of the two types of traffic and compare NOMA with OMA, the more traditional orthogonal partitioning of the radio resources. In the majority of the cases NOMA outperforms OMA being capable to reduce power consumption in practically every condition.}
%\FS{However, optimizing resource and power allocation for the URLLC traffic is a difficult task, taking into account that the closed-form expression of the parallel channel outage probability} can be obtained only in special cases~\cite{Yilmaz2010}. Even approximations of the outage probability as~\cite{Bai2013, Coon2015} are not easy to be addressed in an optimization problem involving non-convex special functions. Moreover, the use of NOMA leads to an even more complex evaluation of outage probability, given by the interference channel~\cite{Wang2018, Li2020}.
%On the other hand, the study made in this paper is, in fact, without loss of generalization if each set of resources is shared with a couple of users having heterogeneous traffic. The results show the overall power consumption of a scenario where users are paired according to the NOMA paradigm. To the best of our knowledge, even this simple scenario has not been investigated with these assumptions; the generalization to more than two users sharing the same group of resources is a future step of this research.

The main contributions of this paper can be summarized as follows:
\begin{itemize}
    \item \mm{We propose a novel approach to downlink  RAN slicing of eMBB and URLLC traffics employing non-orthogonal access technology. Slicing is formulated as a power minimization problem subject to the diverse constraints of the two types of traffic.}
    \item \mm{To address the non-convexity of the optimization problem, we propose a novel heuristic, which follows a hierarchically layered approach:  eMBB traffic is allocated first and then URLLC. This choice is motivated by a practical observation: due to the stringent latency constraints of URLLC, interference cancellation can be applied only at the eMBB receivers, which, accordingly, operate virtually in the absence of interference; for this reason, eMBB resources need to be allocated first.}
    \item \mm{The allocation of URLLC is by far the most challenging of the two allocation problems, being plagued by the presence of the interference of the eMBB traffic and being hindered by the absence of instantaneous CSI. To address this specific problem, we propose two novel algorithms: a) a scheme based on block coordinated descent (BCD) able to converge to the optimal solution and b)  a much lower-complexity strategy based on the use of lookup tables, which obtains results very close to the bound represented by the BCD scheme.}
    \item \mm{We show by numerical simulation that in our setting NOMA outperforms OMA, i.e., that by use of NOMA we can fully exploit the frequency diversity of the system, overcoming some of the shortcomings relative to the latency requirements of the URLLC traffic.}
\end{itemize}

\paragraph*{Paper outline} The remainder of the paper is organized as follows. In Section~\ref{sec:model} the environment and the signal model are described. In the same Section, we focus also on the overall information transmittable and the possible outage events that may occur to the users. In Section~\ref{sec:allocation}, the allocation process is described. We firstly present the minimization power problem; then, we show how to split the problem to obtain low-complexity solutions.
In Section~\ref{sec:algorithms}, we describe the allocation strategy. Results are presented in Section~\ref{sec:results} and conclusions in Section~\ref{sec:conclusions}.

\paragraph*{Mathematical notation} We use capital italic letter to represent sets, e.g. $\mc{A}$, and their cardinality is denoted with capital letter, e.g. $A$. We denote the operator $\mc{A} \setminus \mc{B} = \{x \in \mc{A}, x \notin \mc{B}\}$. Vector are presented in bold uppercase letters $\mb{P}$; the symbol $\succeq$ represent the element-wise comparison between vectors. $\mb{1}$ and $\mb{0}$ are the vectors composed by ones or zeros everywhere, respectively. The symbol $\mc{CN}(a,B)$ represents complex Gaussian distribution with mean value $a$ and variance $B$, while $\text{Exp}(a)$ represents the exponential distribution with mean value $a$.

\section{System Model}
\label{sec:model}
We study a 5G-like single-cell downlink scenario. 
\mm{To capture the dynamics of the slicing problem, which is inherently complex, and to focus our research on the most important factors that affect the coexistence of the two types of traffic, we consider a simplified scenario composed of two users}: an eMBB user, denoted by index $e$ and a URLLC user, denoted by the index $u$. The available resources are organized in a time-frequency grid. In the time domain, we consider a single time slot of duration $T$ [s]. For low-latency communications, the time slot is further divided into a set of mini-slots $\mc{M} = \{0, \dots, M-1\}$, $|\mc{M}| = M$, each one of duration $T_m=T/M$ [s].
We assume that the coherence time of the channel $T_c$ is $T_c \ge T$ so that the channel gains can be assumed constant during an entire time slot.
In the frequency domain, we consider a set $\mc{F} = \{0, \dots, F - 1\}$, $|\mc{F}| = F$ of orthogonal frequency resources, each occupying a bandwidth $\Delta_f$ [Hz].
We refer to a single time-frequency resource as a \emph{mini resource block} (mRB), where the term ``mini'' is used to highlight the reduced dimension in time to a conventional resource block.

\begin{figure}
    \centering
    \includegraphics[width=6.5cm]{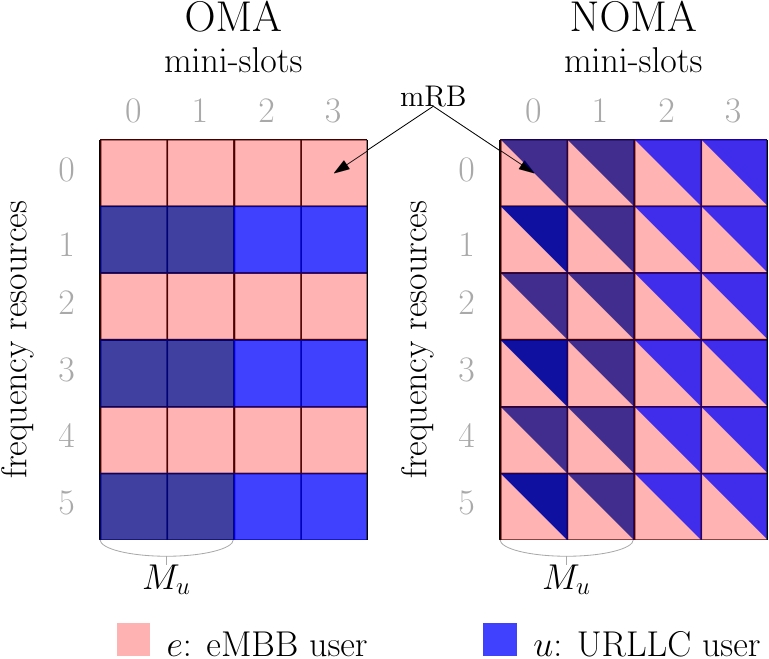}
    \caption{A toy example with a resource grid of $F = 6$ frequency channels and $M = 4$ mini-slots. In the first case,  $F_u = 3$ channels, $\mc{F}_u = \{1, 3, 5\}$, are reserved for URLLC traffic in an OMA paradigm. In the second  case, all the mRBs are reserved for $u$ and $e$, i.e., $F_u = 5$ following a NOMA approach. The URLLC packet is transmitted using the first $M_u = 2$ mini-slots for both OMA and NOMA.}
    \label{fig:resource_grid}
    \vspace{-0.6cm}
\end{figure}

The two users have different objectives and constraints: the eMBB user is modeled as transmitting a pipeline of $N_e$ data bits per slot, while the URLLC user has to meet specific requirements in terms of latency and reliability. In detail, we assume that a packet containing $N_u$ data bits must be delivered within $T_u$ seconds with an outage probability lower than $\epsilon_u$. 
The latency constraints are expressed as a function of the \emph{edge delay}, i.e., the delay between the time at which the message arrives at the transmitter and the time at which the message is effectively transmitted to the user, assuming that all other delay terms have already subtracted from $T_u$~\cite{She2017}. Without loss of generality, the tolerable latency is expressed as a maximum number of mini-slots $M_u^{\max}$.

Resources are assigned to the users based on their traffic type (URLLC or eMMB) and the multiple access technologies adopted (OMA or NOMA). 
Following the 5G NR standard~\cite{3gpp:access}, the URLLC data transmission is allowed to span a certain number of mini-slots, due to the critical time communication. On the other hand, the resource allocation for eMBB transmission can only span the entire slot, in order to support a high communication rate.   

To formalize this concept, we denote as $\mc{F}_u \subseteq \mc{F}$, $|\mc{F}_u| = F_u$, and $\mc{M}_u \subseteq \mc{M}$, $|\mc{M}_u| = M_u$, the sets of frequency and temporal resources allocated for the \emph{transmission} of the data stream of user $u$. 
It is worth noting that the URLLC tolerable latency constraints the transmission time $M_u$, while its $F_u$ resources are reserved for the entire slot, i.e., for $M$ mini-slots. 
%\FS{In other words, if a subset of the mRBs is reserved for the URLLC, this reservation lasts for the whole time slot, i.e., for $M$ mini-slots, while the actual transmission time is yet $M_u \le M$.} 
We remark that this model is a generalization of other models present in the literature, which constrain the URLLC transmission to take place within a single mini-slot, i.e., imposing $M_u = 1$.
When it is compatible with the latency constraints, increasing the time duration of the URLLC data transmission reduces the power needed to transmit the same amount of information. %However, this increases also the latency a URLLC packet can experience if it arrives when the previous packet is still in transmission. 
Without loss of generality, to correctly compare the schemes, we assume that only a single URLLC packet has to be served and it is transmitted immediately after its arrival.% This assumption is without loss of generality if another packet arrives during the transmission and is allocated on a different set of resources.

The subset of the mRBs reserved for $e$ is denoted as $\mc{F}_e\subseteq \mc{F}$, $|\mc{F}_e| = F_e$, while its set of temporal resources comprehend all the $M$ mini-slots, as mention above. %To support high data rate, we exploit all the available time for the transmission of eMBB data stream, resulting in $\mc{M}_e = \mc{M}$. 
%The frequency allocation depends on the multiple access schemes employed.
In the case of OMA, we enforce orthogonal allocation in the frequency domain, multiplexing the users on different spectral resources so that each mRB reserved for $u$ cannot be shared with $e$.
In the case of NOMA, all the resources reserved for $e$ are shared with $u$ employing different values of power to guarantee the traffic requirements of both services. %, as soon as the $u$ transmission does not disrupt the eMBB requirements and vice versa. \FS{Sharing both frequency and time resources, we reserve different values of transmission power to multiplex the two users, aiming to guarantee both the traffic requirements.}
Therefore, the set $\mc{F}_e$ results
\begin{equation}
\label{eq:Fe}
    \mc{F}_e = 
    \begin{cases}
      \mc{F} \setminus \mc{F}_u, \quad &\text{OMA}, \\
      \mc{F}, \quad &\text{NOMA}.
    \end{cases}
\end{equation}
An example of the resource allocation grid for both OMA and NOMA is presented in Fig.~\ref{fig:resource_grid}.

Given the number of resources assigned to each user, and being $N_i$ the number of data bits to be transmitted for user $i\in\{e,u\}$, the average spectral efficiencies per resource [bit/s/Hz] are obtained as
\begin{equation}\label{eq:bit2se}
\begin{aligned}
    r_u &= \frac{N_u}{T_m \Delta_f F_u M_u}, \qquad&
    r_e &= \frac{N_e}{T_m \Delta_f F_e M}.
\end{aligned}
\end{equation}

Finally, the transmitter possesses different knowledge of the channel gains of the two different types of traffic. Instantaneous CSI is not available at the transmitter for URLLC traffic, and, thus, we assume that only the mean signal-to-noise ratio (SNR) $\Gamma_u$ is known.
On the other hand, we assume the complete knowledge of the instantaneous CSI for the eMBB user.

It is worth noting that the same model can be used when URLLC-eMBB coexistence is obtained by puncturing the eMBB data in the frequency domain. In this case, the number of mRBs given to URLLC is always $F_u = F$, while the transmission time is set to $M_u = 1$. With puncturing, the power for the eMBB service must be computed increasing the average spectral efficiency $r_e$, so that the eMBB data can still be recovered after erasing some part of the message. Apart from this detail, the analysis presented in the remainder of the paper is still valid also in the case of puncturing.
% We also remark that the maximum number of recoverable punctured mini-slots $m$ is $M - 1$, to leave at least one mini-slot to the eMBB data transmission. On the other hand, for the general OMA and NOMA schemes presented so far, the URLLC data stream can be transmitted to all $M$ mini-slots. Therefore, a performance comparison between puncturing, OMA, and NOMA is unfair in terms of latency experienced, so we focus on comparing OMA and NOMA only.

\subsection{Signal model}
\label{sec:signalmodel}
We consider a multi-carrier system, where a single mRB $(t, f)$, corresponds to mini-slot $t\in\mc{M}$ and frequency resource $f \in \mc{F}$. Depending on the multiple access technology, each resource can be used simultaneously  by both $e$ and $u$ users (NOMA) or by only one of them (OMA). In any case, the corresponding signals, denoted as $\mb{s}_e(t,f)$ for user $e$ and as  $\mb{s}_u(t,f)$ for user $u$, satisfy the following requirements
\[
\E{||\mb{s}_e(t, f)||^2} = 1, \quad \E{||\mb{s}_u(t, f)||^2} = 1, \quad \E{\mb{s}_e^{\rm H}(t, f) \mb{s}_u(t, f)} = 0.
\]
% We assume different codebooks for URLLC and eMBB users.
We remark that URLLC codewords are exactly contained in a single transmission, spanning contiguous mini-slots.  
In this way, if $M_u$ is chosen appropriately, a single successful transmission carries all the information to the receiver, fulfilling the latency constraint.
On the other hand, the eMBB codewords may span the whole time slot since they do not have any specific latency constraints.
In the case of NOMA, the decoding process implies that one of the two users employs SIC to remove the data stream of the other (interfering) user. However, the cancellation of a user's data stream requires the reception of the entire codeword. Since eMBB codewords span an indeterminate number of mini-slots, the URLLC user that waits for the reception of the whole $e$ codeword may incur a violation of its latency requirements. Therefore, we adopt a NOMA paradigm where \emph{it is always the eMBB user that employs the SIC to remove the interference}, and the URLLC will always be received in the presence of the eMBB interference.

The base station (BS) transmits both eMBB and URLLC data streams using superposition coding. The transmitted signal in an mRB $(t,f)$ is:
\begin{equation}
\begin{aligned}
\mb{x}(t, f) &= \sqrt{P_e(t,f)} \mb{s}_e(t, f) + \sqrt{P_u(t,f)}  \mb{s}_u(t, f)
\end{aligned}
\end{equation}
where $P_e(t,f)$, $P_u(t,f)$ are the power coefficient used to transmit the symbols of $e$ and $u$ on resource $(t,f)$, respectively. 
It is worth noting that this formalization implies the possibility of the transmission of the data stream of a single user $i \in \{u,e\}$ in an OMA fashion by setting the power coefficient of one of the two users equal to 0.  We further denote as $\mb{P}_e$ and $\mb{P}_u$ the vectors collecting all the eMBB and URLLC power coefficients, respectively.

On the receiver side, we can model the signal received by user $i \in \{e, u\}$ as
\begin{equation}
\begin{aligned}
\mb{y}_i(t, f) = h_i(f) \mb{x}(t,f) + \mb{n}_i
\end{aligned}
\end{equation}
where $h_i(f)$, $i \in \{e,u\}$, is the the fading channel gain taking into account both small-scale and large-scale fading, and $n_i \sim \mc{N}(0, \sigma^2 \mb{I}_{n})$, $i \in \{e,u\}$, is the noise at the receiver. Specifically, $h_e(f)$ realization is assumed known at the transmitter, while $h_u(f)$ is assumed unknown.
The channel fading coefficients do not change during the whole slot.
In the remainder of the paper, we assume that fading is Rayleigh distributed; however, this assumption is not essential, and all the considerations can be extended for any other kind of fading.

To simplify the notation, we further denote the normalized instantaneous SNR at the receiver $i\in\{e, u\}$ as $    \gamma_i(f) = \frac{|h_i(f)|^2}{\sigma^2}$,
% \begin{equation}
%     \gamma_i(f) = \frac{|h_i(f)|^2}{\sigma^2}
% \end{equation}
where, for Rayleigh fading, it is $\gamma_i(f) \sim \text{Exp}(\Gamma_i)$  and $\Gamma_i = \mathbb{E}\{\gamma_i\}$, $i\in\{e,u\}$ is the normalized mean SNR at the receiver. Under the assumptions made, $\gamma_e(f)$ is known and $\gamma_u(f)$ is not known at the transmitter. %, with  In the remainder of the paper, we denote as $\Gamma_i = \mathbb{E}\{\gamma_i\}$, $i\in\{e,u\}$ the normalized mean SNR at the receiver.

\subsection{Mutual information}
% Due to the transmission of short packets for the URLLC communication, the model of the communication is accurately addressed by the finite blocklength regime~\cite{Polyanskiy2010}. For this regime, the back-off from the channel capacity is characterized by the so-called channel dispersion, which measures the stochastic variability of the channel.
% Fortunately, in the case of quasi-static channels, the channel dispersion is zero, regardless if the CSI is known \pp{PP: Actually, this is not correct. If the CSI is known, then we practically have a Gaussian channel with some SNR, so the finite blocklength kicks in. I suggest to remove this discussion on the finite blocklength, as it is only introducing confusion. You can say that, for eMBB uses a lot of bits, so we can treat it as infinite, while URLLC works without CSI knowledge and we treat it in the case of outage, for which finite blocklength effects disappear.} or not at the transmitter~\cite{Yang2014}. Hence, the performance of our system may be evaluated by means of outage capacity, even if the blocklength of transmitted data does not go to infinity~\cite{Durisi2016}. 

Due to the high number of informative bits, we can approximate the eMBB traffic to work in the infinite symbol regime, allowing for a Shannon-like information model. Due to the short packets employed for the URLLC communication, the communication is accurately described by the finite blocklength regime when the CSI is known~\cite{Polyanskiy2010}.
\FS{Nevertheless, without the CSI knowledge, the finite blocklength effects disappear, as proven in~\cite{Yang2014, Durisi2016}, and we can address to the outage capacity to model the communication.}
%This means that the main reason for error is the block fading, not the noise.
% To reliably transmit the data stream requested by the users, we must guarantee that the mutual information of both transmissions can sustain the target rate $r_i$, $i\in\{e, u\}$, with high probability.

% According to the previous assumptions, the mutual information of $u$ data stream at receiver $u$ is~\cite{Tse2005}
% \begin{equation} \label{eq:info:u}
%     I_u(\mb{P}_u, \mb{P}_e) = \frac{1}{F_u M_u} \sum_{t\in\mc{M}_u}\sum_{f\in\mc{F}_u} \log_2\left( 1 + \frac{\gamma_u(f) P_u(t, f)}{1 + \gamma_u(f) P_e(t, f)}\right); \quad \text{[bit/s/Hz]}
% \end{equation}
% and the mutual information of $u$ data stream at receiver $e$ is
% \begin{equation} \label{eq:info:ue}
%     I_{u,e}(\mb{P}_u, \mb{P}_e) = \frac{1}{F_u M_u}\sum_{t\in\mc{M}_u}\sum_{f\in\mc{F}_u} \log_2\left( 1 + \frac{\gamma_e(f) P_u(t,f)}{1 + \gamma_e(f) P_e(t,f)}\right). \quad \text{[bit/s/Hz]}
% \end{equation}
% %needed to address the SIC process. 
% Note that $I_{u,e} = 0$ in the OMA case.
% Finally, the mutual information of $e$ data stream at receiver $e$ after a successful SIC process is
% \begin{equation} \label{eq:info:e}
%     I_e(\mb{P}_e) =  \frac{1}{F_e M} \sum_{t\in\mc{M}} \sum_{f\in\mc{F}_e} \log_2\left( 1 + \gamma_e(f) P_e(t,f)\right). \quad \text{[bit/s/Hz]}
% \end{equation}
% where SIC process is always assumed successful in the OMA case, to consider a coherent model for both multiple access paradigms.

According to the previous assumptions, the mutual information of $u$ data stream at receiver $u$ is~\cite{Tse2005}
\begin{equation} \label{eq:info:u}
    I_u(\mb{P}_u, \mb{P}_e) = \frac{1}{F_u M_u} \sum_{t\in\mc{M}_u}\sum_{f\in\mc{F}_u} \log_2\left( 1 + \frac{\gamma_u(f) P_u(t, f)}{1 + \gamma_u(f) P_e(t, f)}\right); \quad \text{[bit/s/Hz]}
\end{equation}
\FS{and the mutual information of $e$ data stream at receiver $e$ after a successful SIC process is
\begin{equation} \label{eq:info:e}
    I_e(\mb{P}_e) =  \frac{1}{F_e M} \sum_{t\in\mc{M}} \sum_{f\in\mc{F}_e} \log_2\left( 1 + \gamma_e(f) P_e(t,f)\right), \quad \text{[bit/s/Hz]}
\end{equation}
where SIC process is always assumed successful in the OMA case.}
\FS{In the NOMA case, the SIC process can fail if the achievable rate for the $u$ data stream at the $e$ receiver is lower than $r_u$. Due to the knowledge of the $e$ channel, the finite blocklength penalty is not zero in this case, and it depends on the required outage probability on the eMBB transmission. On the other hand, due to the relatively high value of eMBB outage probability required (e.g., $10^-2$), the achievable rate is the $90\%$ of the mutual information for every practical blocklength~\cite{Yang2014}. To account for the penalty, we can set $r_{u,e} = r_u / 0.9$ as the target spectral efficiency to sustain for the SIC process to succeed.} 

% SHAME ON ME Hence, we can use the mutual information as a metric, considering a fictional increased of $r_u$, e.g., $r_{u,e} = r_u / 0.9$, to account for the penalty.}
The expression of the mutual information of $u$ data stream at receiver $e$ is
\begin{equation} \label{eq:info:ue}
    I_{u,e}(\mb{P}_u, \mb{P}_e) = \frac{1}{F_u M_u}\sum_{t\in\mc{M}_u}\sum_{f\in\mc{F}_u} \log_2\left( 1 + \frac{\gamma_e(f) P_u(t,f)}{1 + \gamma_e(f) P_e(t,f)}\right), \quad \text{[bit/s/Hz]}
\end{equation}
%needed to address the SIC process. 
where $I_{u,e} = 0$ in the OMA case.
% Finally, the mutual information of $e$ data stream at receiver $e$ after a successful SIC process is
% \begin{equation} \label{eq:info:e}
%     I_e(\mb{P}_e) =  \frac{1}{F_e M} \sum_{t\in\mc{M}} \sum_{f\in\mc{F}_e} \log_2\left( 1 + \gamma_e(f) P_e(t,f)\right). \quad \text{[bit/s/Hz]}
% \end{equation}

% \begin{figure}
%     \centering
%     \input{journalplots/FBL}
%     \caption{Caption}
%     \label{fig:my_label}
% \end{figure}

% We remark that the same instantaneous SNR $\gamma_i(f)$ for the different mini-slots leads to a linear dependence between the mutual information and the number of mini-slots spanned by the transmission.
% If the mutual information transmitting onto $F_i$ frequencies and $M_i$ mini-slots is denoted as $I$, expanding that transmission onto $k M_i$ mini-slots with the same power coefficients will result in a mutual information of $k I$, $\forall k\in\mathbb{N}^{++}$.
% On the other hand, the diversity gain obtained by spanning more frequencies is not easily addressable. To the best of our knowledge, the only concise expression for the diversity gain for parallel channels is given in~\cite{Coon2015}, which depends on the fading distribution. However, the approximation defined in~\cite{Coon2015} is valid only if the same power (signal and inference) is used on all channels; therefore, it cannot be applied in this framework.

\subsection{Outage events}
\label{sec:outages}
Given the mutual information denoted above, it is easy to define the outage events that may occur during the transmission towards both users.

Let us start from the outage events occurring at $e$.
The data stream transmitted to $e$ is incorrectly decoded if: a)  SIC  is not successful, or b) the data stream of $e$ is erroneously decoded after the SIC.
 SIC can not be employed if $\FS{r_{u,e}} > I_{u,e}$, and its probability is
\begin{equation} \label{eq:outage:sic}
    p_{u,e}(\mb{P}_u, \mb{P}_e) = \Pr\{ I_{u,e}(\mb{P}_u, \mb{P}_e) < \FS{r_{u,e}} \}.
\end{equation}
Assuming that the SIC process has been successful, the data stream of $e$ is wrongly decoded at its own receiver if $r_e > I_{e}$, which occurs with probability
\begin{equation} \label{eq:outage:e}
    p_{e}(\mb{P}_e) = \Pr\{ I_{e}(\mb{P}_e) < r_e \}.
\end{equation}
It is worth noting that the assumption of complete knowledge of the CSI for the eMBB user constraints $p_e$ and $p_{u,e}$ to be either 1 or 0.

Let us now focus on the outage event for $u$.
The data stream intended for $u$ is not successfully decoded if a) the URLLC packet is erroneously decoded at receiver $u$, happening if $r_u > I_u$, or b) if the URLLC packet is not entirely received before the latency requirement $M_u^{\max}$.
%Defining the waiting time $W_u$ \pp{PP: This needs to be precisely defined in the model and depicted on the figure.} as the number of mini-slots not used by $u$ between the arrival and the complete transmission of the URLLC packet, 
Assuming the URLLC packet is transmitted as soon as it arrives, the outage probability of $u$ can be formalized as~\cite{Durisi2016}
\begin{equation} \label{eq:outage:u}
\begin{aligned}
p_u(\mb{P}_u, \mb{P}_e) = \Pr\{I_u(\mb{P}_u, \mb{P}_e) \le r_u \cup M_u > M_u^{\max} \}.
\end{aligned}
\end{equation}
Using Boole's inequality, the outage probability can be upper bound by
\begin{equation}\label{eq:outage:u:expanded}
    p_u(\mb{P}_u, \mb{P}_e) \le \Pr\{I_u(\mb{P}_u, \mb{P}_e) \le r_u\} + \Pr\{M_u > M_u^{\max}\}%
\end{equation}
which must be lower than the URLLC reliability constraint $\epsilon_u$.
The probability of the latency term in~\eqref{eq:outage:u:expanded} can be easily constrained to be 0 by choosing a reasonable value of $M_u \le M_u^{\max}$ so that it is 
\begin{equation}\label{eq:outage:u:bounded}
    p_u(\mb{P}_u, \mb{P}_e) \le \Pr\{I_u(\mb{P}_u, \mb{P}_e) \le r_u\} \le \epsilon_u.
\end{equation}
%Since channel gains are static on the time dimension, the set $\mc{M}_u$ can be chosen randomly between all the possible sets guaranteeing relation~\eqref{eq:M}.
According to~\eqref{eq:bit2se}, spreading the same number of informative bits on more than one mini-slot directly reduces the target rate $r_u$. Hence, a way to reduce the transmission power is to exploit the tolerable delay as much as possible. Nevertheless, using a large value of $M_u$ for the transmission of a packet may prevent the prompt transmission of the following URLLC packet, increasing the latency experienced by the latter. Hence, the complete design of the transmission time should also take into account the URLLC traffic model, which is beyond the scope of this investigation. In the remainder, we address the problem of minimum power assuming $M_u$ is given. %, guaranteeing that the obtained mutual information $I_u$ can sustain the rate $r_u$ with probability $1 - \epsilon_u$.
% Therefore, we allocate the power coefficients to guarantee that the mutual information $I_u$ can sustain the rate $r_u$ with probability $1 - \epsilon_u$. % In the remainder of the paper we consider that the $M_u$ is chosen following~\eqref{eq:M}, thus the URLLC outage happens if the reliability requirement is violated.

\section{Resource allocation}
\label{sec:allocation}
We can now formalize the overall minimum power allocation problems for both OMA and NOMA schemes. We further denote as $P^\text{tot} = \sum_{t\in\mc{M}}\sum_{f\in\mc{F}} P_u(t, f) + P_e(t,f)$ the overall power spent.

The OMA power allocation problem is the following
\begin{align}
\label{op:oma}
    \min_{\mb{P}_e, \mb{P}_u}& P^\text{tot} \\
    \text{s.t. } & p_{e}(\mb{P}_e) = 0, \tag{\ref{op:oma}.a} \label{op:oma:pe} \\
    & P_u(t,f) P_e(t,f) = 0, \, \forall t\in\mc{M}, \forall f\in\mc{F} \tag{\ref{op:oma}.b} \label{op:oma:pupe}\\
    & p_u(\mb{P}_u,\mb{P}_e=0)\le \epsilon_u, \tag{\ref{op:oma}.c} \label{op:oma:ou}\\
    & \mb{P}_u \succeq 0, \mb{P}_e \succeq 0, \tag{\ref{op:oma}.d}
\end{align}
where constraint~\eqref{op:oma:pe} assures that the eMBB is perfectly decoded, constraint~\eqref{op:oma:pupe} enforces the orthogonality of the allocation, and constraint~\eqref{op:oma:ou}  guarantees the reliable transmission of URLLC traffic as in \eqref{eq:outage:u:bounded}, under the condition that $\mb{P}_e=0$ due to orthogonal allocation.
The NOMA allocation problem is almost identical, except for the orthogonality constraint~\eqref{op:oma:pupe} substituted by the SIC requirement~\eqref{eq:outage:sic}, which requires that SIC is successful at the eMBB receiver. Hence, the minimization problem results
\begin{align}
\label{op:noma}
    \min_{\mb{P}_e, \mb{P}_u}& P^\text{tot} \\
    \text{s.t. } & p_{e}(\mb{P}_e) = 0, \tag{\ref{op:noma}.a}\label{op:noma:pe}\\
    & p_{u,e}(\mb{P}_u, \mb{P}_e) = 0, \label{op:noma:ue} \tag{\ref{op:noma}.b}\\
    & p_u(\mb{P}_u, \mb{P}_e) \le \epsilon_u, \tag{\ref{op:noma}.c} \label{op:noma:ou}\\
    & \mb{P}_u \succeq 0, \mb{P}_e \succeq 0, \tag{\ref{op:noma}.d}
\end{align}
\mm{Both problems are not convex w.r.t. $\mb{P}_u$ and $\mb{P}_e$. In both cases, the main challenge is represented by the constraints~\eqref{op:oma:ou} and~\eqref{op:noma:ou} for which no known solution exists in closed form, and only bounds or approximations can be used, e.g., see~\cite{Bai2013, Coon2015, Li2020}. Unfortunately, these approximations are not convex, and special functions are involved, so the optimization of the power coefficients is a hard task, even harder in the case of NOMA due to the presence of the interference. Therefore, for the NOMA scheme, where the two types of traffic interfere, we propose a heuristic based on a \emph{layered approach} designed to decouple the allocations of the two-class of users: we first find the optimal power distribution for the eMBB and then for the URLLC users. The reason for first optimizing the eMBB allocation is the perfect knowledge of the $e$ channel, and its transmissions are virtually interference-free, also in the NOMA scenario, since the interference can be canceled with a very high probability of success. We will show that constraints~\eqref{op:oma:ou} and~\eqref{op:noma:ou} can be addressed by numerical methods based on the use of look-up tables once the level of the eMBB interference is known.}

\FS{We remark that practical implementation of the system should account for the maximum transmission power of the BS, which may yield an empty feasible set, i.e., no feasible solution can be found.
Being interested in comparing OMA and NOMA in terms of minimum power consumed, we do not consider such power limit in problems~\eqref{op:oma} and~\eqref{op:noma}, but we present experimentally the effect of imposing a power budget in Section~\ref{sec:results}.}
% posing a further constraint in the definition of the optimization problems. Being interested in comparing OMA and NOMA in terms of minimum power consumed, we do not consider such power limit in problems~\eqref{op:oma} and~\eqref{op:noma}, which may lead to an empty feasible set, i.e., no feasible solution can be found. We will still present the effect of imposing a power budget in Section~\ref{sec:results}.}

%In particular, we minimize the eMBB power subject to its decoding process~\eqref{eq:outage:e}; then, we found the minimum URLLC power guaranteeing SIC process~\eqref{eq:outage:sic}; finally, we evaluate the URLLC power allocation constrained to~\eqref{eq:outage:u}.
% Having considered a single coherence time, the power analysis can be made assuming that $P_u(t,f)$ and $P_e(t,f)$ will not change on different mini-slots.
% Moreover, given~\eqref{eq:M}, the time resources of URLLC transmissions are already defined. Thus, we can impose $P_i(t,f) = P_i(f)$, $\forall t \in \mc{M}_i$, $i \in\{e,u\}$.
%In the following Subsections, the decoupled problems are presented. Then, in Section~\ref{sec:algorithms}, we describe the overall procedure for the power allocation.

\subsection{Time-frequency selection}
\label{sec:allocation:e}
In the time domain, the eMBB transmission is allocated to an entire slot, as customary. Even if the URLLC could transmit on the whole time slot, its actual transmission time is constrained by the tolerable latency $M_u^\text{max} \ge M_u$, obtaining~\eqref{eq:outage:u:bounded} from~\eqref{eq:outage:u:expanded}.
In the frequency domain, the allocation depends on the access method selected. In the case of NOMA, all the resources are shared. Hence, the allocation sets are $\mc{F}_u = \mc{F}_e = \mc{F}$.
In the case of OMA, we assume that the number of frequencies spanned by the URLLC transmission $F_u$ has been selected in some way\footnote{In Section~\ref{sec:results} different choices of the number of frequencies are compared.}, resulting in knowing exactly the number of resources $F_u$ and $F_e$ reserved for both users.
Therefore, we evaluate the frequency sets as follows. 
We remark that the mutual information $I_e$ defined in~\eqref{eq:info:e} depends only on the power coefficients $\mb{P}_e$ for both OMA and NOMA cases; therefore, it is maximized if the set $\mc{F}_e$ contains the frequencies with the highest instantaneous SNR $\gamma_e(f)$ among the possible ones, or, equivalently, if the set $F_u$ contains the mRB with lowest channel gains.
Hence, the set of frequency resources given to $u$ is
\begin{equation}
    \label{eq:Fu}
    \mc{F}_u = \argmin_{\mc{F}_u' \subseteq \mc{F}, |\mc{F}_u'| = F_u} \sum_{f\in\mc{F}_u'} \gamma_e(f),
\end{equation}
and the set $\mc{F}_e$ can be computed following~\eqref{eq:Fe}. We remark that the mRB selection given in~\eqref{eq:Fu} does not influence the URLLC transmission, not knowing the CSI for those channels.

\subsection{eMBB allocation}
Since we have assumed that the channel is constant over the entire slot, the power allocation for eMBB, which is not affected by any interference, will be the same in every mini-slot. More formally, $P_e(t,f) = P_e(f)$, $\forall t\in\mc{M}$, $\forall f\in\mc{F}$, and~\eqref{eq:info:e} can be simplified in:
\begin{equation}
    I_e(\mb{P}_e) =  \frac{1}{F_e}\sum_{f\in\mc{F}_e} \log_2\left( 1 + \gamma_e(f) P_e(f)\right).
\end{equation}

Following the above assumptions, we can obtain the minimum value of eMBB power by solving 
\begin{equation}
\label{eq:op:e}
\begin{aligned}
    \min_{\mb{P}_e \succeq 0} &\left\{\sum_{f\in\mc{F}_e} P_e(f) \, \big| \, 
    %\text{s.t. } & 
    \sum_{f\in\mc{F}_e} \log_2\left( 1 + \gamma_e(f) P_e(f)\right) \ge F_e r_e \right\}. 
\end{aligned}
\end{equation}
The solution of problem~\eqref{eq:op:e} can be computed through the well-known water-filling approach~\cite{Tse2005}. \FS{Coming from the general water-filling formulation~\cite{Palomar2005, He2013}, it is easy to derive the solution providing the minimum power spent. %The derivation of the closed-form formulation can be seen in~\cite{saggese2022phdthesis}. 
The optimal power $\mb{P}_e$ results}
\begin{equation}
 \label{eq:power:e}
 P_e(f) = 
 \begin{cases}
 \displaystyle
 2^{r_e \frac{F_e}{F_e^+}} \prod_{i\in\mc{F}_e^+} \left(\frac{1}{\gamma_e(i)}\right)^{\frac{1}{F_e^+}} - \frac{1}{\gamma_e(f)}, \quad  f\in\mc{F}_e^+\\    
 0, \quad \text{otherwise}
 \end{cases}
\end{equation}
where the set $\mc{F}_e^+ \subseteq \mc{F}_e$, $|\mc{F}_e^+| = F_e^+$, is composed by all and only the mRBs guaranteeing $P_e(f) > 0$, $\forall f\in\mc{F}_e^+$~\cite{He2013} % saggese2022phdthesis}. 

Similarly, %following the definition~\eqref{eq:info:ue},
we can obtain the minimum power spent to transmit the URLLC packet satisfying the SIC requirement~\eqref{op:noma:ue}. The optimization problem to solve is
\begin{equation}
\label{eq:op:sic}
\begin{aligned}
    \min_{\mb{P}_u \succeq 0} &\left\{\sum_{f\in\mc{F}_u} P_u(f) \, \big| \, %\\
    %\text{s.t. } &
    \sum_{f\in\mc{F}_u} \log_2\left(1 + \frac{\gamma_e(f) P_u(f)}{1 + \gamma_e(f) P_e(f)}\right) \ge F_u \FS{r_{u,e}}\right\}.
\end{aligned}
\end{equation}
Also in this case, a water-filling approach may be employed considering that the normalized channels gains involved in the optimization are $\frac{\gamma_e(f)}{1 + \gamma_e(f) P_e(f)}$. The solution obtained is labeled as $\mb{P}_u^\text{SIC}$ to highlight the fact that this is the minimum power needed to satisfy the SIC constraint, whose components are
\begin{equation} \label{eq:power:sic}
    P_u^\text{SIC}(f) =
    \begin{cases}
    \displaystyle
       2^{\FS{r_{u,e}} \frac{F_u}{F_{u}^+}} \prod_{i\in\mc{F}_{u}^+} \left(\frac{1}{\gamma_e(i)} + P_e(i)\right)^{\frac{1}{F_{u}^+}} - \frac{1}{\gamma_e(f)} - P_e(f), \quad f\in\mc{F}_{u}^+, \\
       0, \quad \text{otherwise}
    \end{cases}
\end{equation}
where $\mc{F}_{u}^+ \subseteq \mc{F}_u$, $|\mc{F}_{u}^+| = F_{u}^+$, collects all and only the mRBs reserved for $u$ guaranteeing $P_u^\text{SIC}(f) > 0$. When we implement the OMA paradigm, we impose $\mc{F}_{u}^+ = \emptyset$, resulting in $\mb{P}_u^\text{SIC} = \mb{0}$. 

Guaranteeing the success of SIC at the eMBB receiver, the power coefficients obtained in~\eqref{eq:power:sic} represent the minimum value of power f $u$ transmissions. 
%$P_u^\text{SIC}(f)$ depends on the set $\mc{F}_{u}^+$: the higher is the cardinality $F_{u}$, the higher will be $F_u^+$, and the lower is the power spent.
Moreover, we remark that better eMBB channel gain conditions lead to a lower power consumption for SIC due to higher $\gamma_e(f)$ and lower $P_e(f)$ given by~\eqref{eq:power:e}. On the other hand, when the eMBB user experiences poor channel conditions, the SIC process may be the dominant effect in the allocation of URLLC power, as we show in Section~\ref{sec:results}.

\subsection{URLLC allocation}
\label{sec:allocation:u}
We now address the issue of optimizing the power allocation of the URLLC user under the constraint 
%Therefore eq.~\eqref{eq:outage:u} can be rewritten as
\begin{equation}
    \label{eq:outage:ufull}
    p_u(\mb{P}_u;\mb{P}_e) = \Pr \left\{ I_u(\mb{P}_u;\mb{P}_e) \le r_u \right\} \le \epsilon_u,
\end{equation}
where $\mb{P}_e$ is now a parameter and not an optimization variable since we assume that the power coefficients for the eMBB user have already been evaluated by solution~\eqref{eq:power:e}.
As in the eMBB case, we can obtain all the power coefficients by studying a single mini-slot. The eMBB interference power coefficients depend on the frequency only, being the same for different mini-slots. Having the same interference and channel gain - even if unknown -
the power coefficient for URLLC transmission will depend only on the gain of the specific frequency channel, i.e. $P_u(t,f) = P_u(f)$,  $t\in\mc{M}_u$, $f\in\mc{F}_u$.
Thus, the mutual information $I_u$ can be simplified in:
\begin{equation} \label{eq:info:u:simple}
    I_u(\mb{P}_u;\mb{P}_e) = \frac{1}{F_u}\sum_{f\in\mc{F}_u} \log_2\left( 1 + \frac{\gamma_u(f)P_u(f)}{1 + \gamma_u(f) P_e(f)} \right).
\end{equation}

Taking into account both the constraints on the SIC process and the reliability requirement, we can now express the URLLC allocation problem as
\begin{equation}
    \label{eq:op:u}
    \begin{aligned}
    \min_{\mb{P}_u \succeq \mb{P}_u^\text{SIC}} &\sum_{f\in\mc{F}_u} P_u(f), \\
    \text{s.t. } & \Pr \left\{\sum_{f\in\mc{F}_u} \log_2\left( 1 + \frac{\gamma_u(f)P_u(f)}{1 + \gamma_u(f) P_e(f)} \right) \le F_u r_u \right\} \le \epsilon_u.
\end{aligned}
\end{equation}
In theory, the minimization of the objective function of~\eqref{eq:op:u} may be obtained by applying a gradient descent approach if the chosen power coefficients always lie in the feasible set.
Nevertheless, due to the unknown formulation of the outage probability, projection~\cite{MAPEL} or conditional gradient descent algorithms~\cite{Jaggi2013} cannot be applied directly.

In the following, we show an analytic solution of problem~\eqref{eq:op:u} for two particular cases: single frequency and interference-limited scenario. % -, and we propose two different algorithms able to solve the minimization problem in a general case.

\subsubsection{Single frequency resource}
\label{sec:allocation:u:single}
When a single frequency resource is allocated, i.e. $F_u = 1$, the solution of~\eqref{eq:outage:ufull} is known~\cite{Wang2018}, and the minimum transmission power can be obtained imposing $p_u = \epsilon_u$. For a Rayleigh fading, we obtain the minimum power value on the only active frequency $f\in\mc{F}_u$ as
\begin{equation}\label{eq:power:u_single}
    P_u^\text{single}(f) = (2^{r_u} - 1) \left( P_e(f) - \frac{1}{\Gamma_u \ln(1 - \epsilon_u)} \right).
\end{equation}
% \begin{equation}\label{eq:power:u_single}
%     P_u^\text{single}(f) = 
%     \begin{cases}
%     \displaystyle
%     (2^{r_u} - 1) \left( P_e(f) - \frac{1}{\Gamma_u \ln(1 - \epsilon_u)} \right), \quad f\in\mc{F}_u, \\
%     0, \quad \text{otherwise}.
%     \end{cases}
% \end{equation}
In Fig.~\ref{fig:urllc_outage:single}, we show $P_u^\text{single}(f)$ as a function of the mean normalized SNR $\Gamma_u$. The reliability is set as $\epsilon_u = 10^{-5}$ and $P_e(f) = 0$ dBm. For example, if $r_u = 1/3$, to have $P_u^\text{single}(f) \le 10$ dBm we need $\Gamma_u \ge 61$ dB. Hence, transmitting on a single frequency can work only when the user experiences very good channel conditions. %, due to close distance from the BS or favorable scattering scenario.
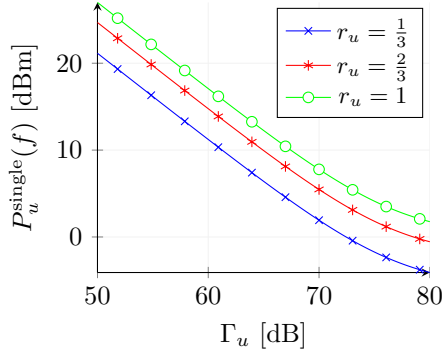
\begin{figure}[bt]
    \centering
    \begin{tikzpicture}
    \begin{axis}[width=6cm, 
                 grid=both,
                 axis lines = left,
                 grid style={line width=.1pt, draw=gray!10},
                 xlabel = $\Gamma_u$ {[dB]}, 
                 %ymax=10^(-1), 
                 ylabel = {$P_u^\text{single}(f)$ [dBm]},
                 legend style ={font=\footnotesize},]
        % First plot
        \addplot [domain=50:80, samples=100, color=blue, mark=x, mark repeat=10, mark phase=7]{10*log10((2^(1/3) - 1) * ( 10^-3 - 1 / 10^(x/10)/ ln(1 - 10^(-5)))) + 30};
        \addlegendentry{$r_u = \frac{1}{3}$};
        \addplot [domain=50:80, samples=100, color=red, mark=asterisk, mark repeat=10, mark phase=7]{10*log10( (2^(2/3) - 1) * ( 10^-3 - 1 / 10^(x/10)/ ln(1 - 10^(-5))) ) + 30};
        \addlegendentry{$r_u = \frac{2}{3}$};
        \addplot [domain=50:80, samples=100, color=green, mark=*, mark repeat=10, mark phase=7, mark options={fill=white}]{10*log10( (2^1 - 1) * ( 10^-3 - 1 / 10^(x/10)/ ln(1 - 10^(-5)))) + 30};
        \addlegendentry{$r_u = 1$};
    \end{axis}
    \end{tikzpicture}
    \caption{$P_u^\text{single}(f)$ as a function of $\Gamma_u$ according to~\eqref{eq:power:u_single}, when $\epsilon_u = 10^{-5}$ and $P_e(f) = 0$ dBm.}
    \label{fig:urllc_outage:single}
\end{figure}
When $F_u > 1$, the exact evaluation of the probability presented in~\eqref{eq:outage:ufull} is not known in a closed-form. Experimental results,   shown in Fig.~\ref{fig:outage:experimental}, obtained by Monte Carlo simulations,  prove that increasing the available frequency resources will lead to a great improvement in terms of reliability. In other words, in a general case when the channel conditions are not good enough, we must rely on the diversity gain. % provided by using multiple frequencies.

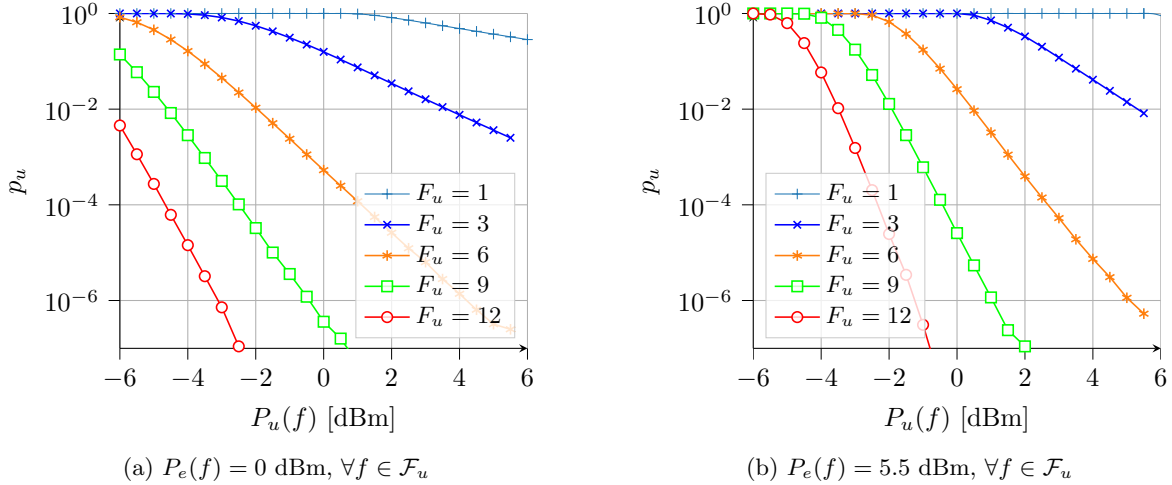
\begin{figure}
    \begin{subfigure}[c]{.49\columnwidth}
        \centering
        % This file was created by tikzplotlib v0.9.8.
\begin{tikzpicture}

\definecolor{color0}{rgb}{0.12156862745098,0.466666666666667,0.705882352941177}
\definecolor{color1}{rgb}{1,0.498039215686275,0.0549019607843137}
\definecolor{color2}{rgb}{0.172549019607843,0.627450980392157,0.172549019607843}
\definecolor{color3}{rgb}{0.83921568627451,0.152941176470588,0.156862745098039}

\begin{axis}[
legend cell align={left},
legend style={fill opacity=0.8, draw opacity=1, text opacity=1, draw=white!80!black, at={(0.97,0.03)},
  anchor=south east},
log basis y={10},
tick align=outside,
tick pos=left,
x grid style={white!69.0196078431373!black},
xlabel={\(\displaystyle P_u(f)\) [dBm]},
xmajorgrids,
xmin=-6, xmax=6,
xtick style={color=black},
y grid style={white!69.0196078431373!black},
ylabel={\(\displaystyle p_u\)},
ymajorgrids,
ymin=1e-07, ymax=1,
ymode=log,
ytick style={color=black}
]
\addplot [domain=-6:0, samples=73, color=color0, mark=+, mark repeat=6, mark phase=0, mark options={fill=white}, forget plot]{1};
\addplot [domain=0:6, samples=73, color=color0, mark=+, mark repeat=6, mark phase=0, mark options={fill=white}]{1 - e^(- (1/1000) / (10^(x/10 - 3) - 10^(-3)) ) };
\addlegendentry{$F_u = 1$};
\addplot [semithick, blue, mark=x]
table {%
-6 1
-5.5 1
-5 0.99998911
-4.5 0.99794821
-4 0.97873861
-3.5 0.9218498
-3 0.82307522
-2.5 0.69369117
-2 0.55358807
-1.5 0.42269238
-1 0.31182948
-0.5 0.22394584
0 0.15743425
0.5 0.10900816
1 0.07458592
1.5 0.05084165
2 0.0345731
2.5 0.02359259
3 0.01612804
3.5 0.01105393
4 0.00760317
4.5 0.0052534
5 0.0036354
5.5 0.0025168
};
\addlegendentry{$F_u=3$}
\addplot [semithick, color1, mark=asterisk, mark options={fill=white}]
table {%
-6 0.83463771
-5.5 0.65405312
-5 0.4550131
-4.5 0.28505001
-4 0.16380593
-3.5 0.0880479
-3 0.04501722
-2.5 0.0222512
-2 0.01071281
-1.5 0.005096
-1 0.00240396
-0.5 0.00112858
0 0.00053074
0.5 0.00025091
1 0.00011903
1.5 5.568e-05
2 2.608e-05
2.5 1.236e-05
3 6.33e-06
3.5 2.81e-06
4 1.38e-06
4.5 6.4e-07
5 3.2e-07
5.5 2.5e-07
};
\addlegendentry{$F_u=6$}
\addplot [semithick, green, mark=square*, mark options={solid, fill=white}]
table {%
-6 0.13883803
-5.5 0.05916983
-5 0.02293734
-4.5 0.00830285
-4 0.00285441
-3.5 0.00095355
-3 0.00031574
-2.5 0.00010218
-2 3.266e-05
-1.5 1.006e-05
-1 3.58e-06
-0.5 1.2e-06
0 3.6e-07
0.5 1.6e-07
1 5e-08
1.5 0
2 0
2.5 0
3 0
3.5 0
4 0
4.5 0
5 0
5.5 0
};
\addlegendentry{$F_u=9$}
\addplot [semithick, red, mark=*, mark options={solid, fill=white}]
table {%
-6 0.00455385
-5.5 0.00114086
-5 0.00027334
-4.5 6.175e-05
-4 1.434e-05
-3.5 3.22e-06
-3 7.2e-07
-2.5 1.1e-07
-2 4e-08
-1.5 0
-1 0
-0.5 0
0 0
0.5 0
1 0
1.5 0
2 0
2.5 0
3 0
3.5 0
4 0
4.5 0
5 0
5.5 0
};
\addlegendentry{$F_u=12$}
\end{axis}

\end{tikzpicture}
        \caption{$P_e(f) = 0$ dBm, $\forall f \in \mc{F}_u$}
    % \vspace{0.2cm}
    \end{subfigure}
    ~
    \begin{subfigure}[c]{.49\columnwidth}
        \centering
        % This file was created by tikzplotlib v0.9.8.
\begin{tikzpicture}

\definecolor{color0}{rgb}{0.12156862745098,0.466666666666667,0.705882352941177}
\definecolor{color1}{rgb}{1,0.498039215686275,0.0549019607843137}
\definecolor{color2}{rgb}{0.172549019607843,0.627450980392157,0.172549019607843}
\definecolor{color3}{rgb}{0.83921568627451,0.152941176470588,0.156862745098039}

\begin{axis}[
legend cell align={left},
legend style={
  fill opacity=0.8,
  draw opacity=1,
  text opacity=1,
  at={(0.03,0.03)},
  anchor=south west,
  draw=white!80!black
},
log basis y={10},
tick align=outside,
tick pos=left,
x grid style={white!69.0196078431373!black},
xlabel={\(\displaystyle P_u(f)\) [dBm]},
xmajorgrids,
xmin=-6, xmax=6,
xtick style={color=black},
y grid style={white!69.0196078431373!black},
ylabel={\(\displaystyle p_u\)},
ymajorgrids,
ymin=1e-07, ymax=1,
ymode=log,
ytick style={color=black}
]
\addplot [domain=-6:5.5, samples=139, color=color0, mark=+, mark repeat=6, mark phase=0, mark options={fill=white}, forget plot]{1};
\addplot [domain=5.5:12, samples=72, color=color0, mark=+, mark repeat=6, mark phase=0, mark options={fill=white}]{1 - e^(- (1/1000) / (10^(x/10 - 3) - 10^(5.5/10 - 3)) ) };
\addlegendentry{$F_u = 1$};
\addplot [semithick, blue, mark=x]
table {%
-6 1
-5.5 1
-5 1
-4.5 1
-4 1
-3.5 1
-3 1
-2.5 1
-2 1
-1.5 1
-1 1
-0.5 1
0 0.99966538
0.5 0.92225788
1 0.7159558
1.5 0.50472021
2 0.33116013
2.5 0.20182846
3 0.11928696
3.5 0.07015622
4 0.04122431
4.5 0.02415075
5 0.01409681
5.5 0.00821299
};
\addlegendentry{$F_u=3$}
\addplot [semithick, color1, mark=asterisk]
table {%
-6 1
-5.5 1
-5 1
-4.5 1
-4 1
-3.5 1
-3 0.99881761
-2.5 0.92096535
-2 0.66831576
-1.5 0.37547836
-1 0.17349416
-0.5 0.07023342
0 0.02616017
0.5 0.00931205
1 0.00325024
1.5 0.00111937
2 0.00039402
2.5 0.00014208
3 5.284e-05
3.5 1.902e-05
4 7.39e-06
4.5 3.08e-06
5 1.14e-06
5.5 5.3e-07
};
\addlegendentry{$F_u=6$}
\addplot [semithick, green, mark=square*, mark options={fill=white}]
table {%
-6 1
-5.5 1
-5 0.99999938
-4.5 0.9879599
-4 0.81344055
-3.5 0.45399773
-3 0.17519119
-2.5 0.05175687
-2 0.01284566
-1.5 0.00286105
-1 0.00060777
-0.5 0.00012758
0 2.564e-05
0.5 5.42e-06
1 1.16e-06
1.5 2.4e-07
2 1.1e-07
2.5 1e-08
3 1e-08
3.5 0
4 0
4.5 0
5 0
5.5 0
};
\addlegendentry{$F_u=9$}
\addplot [semithick, red, mark=*, mark options={fill=white}]
table {%
-6 0.99982164
-5.5 0.94865083
-5 0.62778786
-4.5 0.24026776
-4 0.05845437
-3.5 0.01046614
-3 0.00154208
-2.5 0.00020264
-2 2.455e-05
-1.5 3.44e-06
-1 3.1e-07
-0.5 2e-08
0 0
0.5 0
1 0
1.5 0
2 0
2.5 0
3 0
3.5 0
4 0
4.5 0
5 0
5.5 0
};
\addlegendentry{$F_u=12$}
\end{axis}

\end{tikzpicture}
        \caption{$P_e(f) = 5.5$ dBm, $\forall f \in \mc{F}_u$}
    \end{subfigure}
    \caption{Experimental results of $p_u$~\eqref{eq:outage:ufull} versus $P_u(f)$ with different number of $F_u$, $F_u r_u = 1$, $\Gamma_u = 30$ dB and $P_u(f) = P_u$, $\forall f \in \mc{F}_u$.}
    \label{fig:outage:experimental}
\end{figure}

\subsubsection{Interference-limited scenario}
\label{sec:allocation:u:IL}
A bound of the power consumption for URLLC in the NOMA case can be found exploiting the asymptotic behavior with respect to $\gamma_u(f)$. 
When the URLLC channel conditions are very good, we can approximate each term of~\eqref{eq:info:u:simple} depending on the value of the interference.
In detail, when the $u$ channel gain is large, we can approximate the mutual information  as
\begin{equation} \label{eq:info:approx1}
    \log_2\left(1 + \frac{\gamma_u(f) P_u(f)}{1+\gamma_u(f) P_{e}(f)} \right) \xlongrightarrow{\Gamma_u \rightarrow \infty} \begin{cases}
    \log_2\left(1 + \frac{P_u(f)}{P_e(f)} \right)&\mbox{if }P_e(f) > 0\\
    \log_2\left(\gamma_u(f) P_u(f) \right)&\mbox{if }P_e(f) = 0.
    \end{cases}
\end{equation}
%On the other hand, considering an mRB $g\in\mc{F}_u$ where the interference is null, i.e. $P_e(g) = 0$, the contribution of its term in the summation can be approximated by
%\begin{equation} \label{eq:info:approx2}
%\log_2\left(1 + \gamma_u(f) P_u(f) \right) \xlongrightarrow{\Gamma_u \rightarrow \infty} %\log_2\left(\gamma_u(f) P_u(f) \right).
%\end{equation}
Using the set definition given in Section~\ref{sec:allocation:e}, the mRBs having $P_e(f) > 0$ are collected into the set $\mc{F}_u \cap \mc{F}_e^+$, i.e., the URLLC reserved frequencies shared with $\mc{F}_e^+$; the mRBs having $P_e(f) = 0$ are collected in the set $\mc{F}_u \setminus \mc{F}_e^+$, i.e., the URLLC reserved frequency resources which are not in $\mc{F}_e^+$.
Therefore, \FS{for high values of $\Gamma_u$,} equation~\eqref{eq:info:u:simple} can be well approximated by the following
\begin{equation} \label{eq:info:u:IL:complete}
    F_u I_u(\mb{P}_u;\mb{P}_e) \approx 
    \sum_{f\in\mc{F}_u\cap\mc{F}_e^+}
    \log_2 \left( 1 + \frac{P_u(f)}{P_e(f)} \right) + 
    \sum_{f\in\mc{F}_u \setminus\mc{F}_e^+} \log_2 \left( \gamma_u(f) P_u(f) \right).
\end{equation}
\FS{Experimental results\footnote{\FS{The results are easily reproducible from expression~\eqref{eq:info:u:IL:complete}, but not provided for lack of space.}} show that the approximation is tight for $\Gamma_u > 55$ dB.}

The summation of terms involving no interference leads to a formulation of the outage probability, which can be tightly bounded using the approach given in~\cite{Bai2013}. However, this bound is still very complicated depending on special functions, while we are interested in a simple formulation suggesting the behavior of the power allocation for the interference-limited case. 
To overcome this, we note that the term where $P_e(f) > 0$ has a high probability of being greater than the biggest term where $P_e(f) = 0$. In other words, if we consider the mRB with the lowest interference greater than 0, i.e., $f^{\min} = \argmin_{f\in\mc{F}_u\cap\mc{F}_e^+} P_e(f)$,
% \begin{equation} \label{eq:fmin}
% f^{\min} = \argmin_{f\in\mc{F}_u\cap\mc{F}_e^+} P_e(f),
% \end{equation}
the probability that each realization of the SNR in $\mc{F}_u \setminus \mc{F}_e^+$ is greater than the lowest interference is
\begin{equation}
    \prod_{f\in\mc{F}_u \setminus \mc{F}_e^+}\Pr\left\{ \gamma_u(f) \ge \frac{1}{P_e(f^{\min})} \right\} = e^{-\frac{|\mc{F}_u \setminus \mc{F}_e^+|}{\Gamma_u P_e(f^{\min})}} \xlongrightarrow{\Gamma_u \rightarrow \infty} 1,
\end{equation}
for i.i.d. Rayleigh fading.
Therefore, we can substitute $1/P_e(f^{\min})$ to all the channels experiencing no interference $\mc{F}_u \setminus \mc{F}_e^+$. In this way, we obtain an approximate formulation which is a lower bound of~\eqref{eq:info:u:IL:complete} almost surely.  

To summarize, when the system is interference-limited, equation~\eqref{eq:info:u:simple} is bounded with high probability by
\begin{equation} \label{eq:info:u:IL}
    I_u(\mb{P}_u;\mb{P}_e) \ge \tilde{I_u}(\mb{P}_u;\mb{P}_e) = \frac{1}{F_u} \sum_{f\in\mc{F}_u} \log_2 \left( 1 + \frac{P_u(f)}{P_e(f)} \right),
\end{equation}
where it is implied that $P_e(f) = P_e(f^{\min})$, $\forall f\in\mc{F}_u \setminus \mc{F}_e^+$. Hence, finding the power coefficients $\mb{P}_u$ guaranteeing $\tilde{I_u}(\mb{P}_u;\mb{P}_e) \ge r_u$, we obtain that $\Pr\{I_u(\mb{P}_u);\mb{P}_e \ge r_u\} \rightarrow 1$.

Using~\eqref{eq:info:u:IL}, the minimum power coefficients can be found solving the following problem
\begin{equation}
\label{eq:op:interferencelimited}
    \begin{aligned}
    \min_{\mb{P}_u \succeq 0} &\left\{\sum_{f\in\mc{F}_u} P_u(f) \, \big| \,
    %\text{s.t. } & 
    \sum_{f\in\mc{F}_u} \log_2\left(1 + \frac{P_u(f)}{P_e(f)}\right) \ge F_u r_u \right\}.
\end{aligned}
\end{equation}
Using the water-filling approach, setting the channel gain of each mRB as $1/P_e(f)$, we obtain
\begin{equation}
\label{eq:power:u:IL}
    P_u^\text{IL}(f) = 
    \begin{cases}
    \displaystyle
    2^{r_u F_u} \prod_{i\in\mc{F}_{u}} P_e(i)^{\frac{1}{F_{u}}} - P_e(f), \quad f\in\mc{F}_{u}, \\
    0, \quad \text{otherwise}.
    \end{cases}
\end{equation}
Solution~\eqref{eq:power:u:IL} leads to a vector $\mb{P}_u^\text{IL}$ which guarantees that $\tilde{I_u}(\mb{P}_u^\text{IL};\mb{P}_e) \ge r_u$ for the interference-limited scenario. 
In the numerical results, we show that the interference-limited scenario may occur when the eMBB user is the farther user. In this kind of scenarios, $\mb{P}_u^\text{IL}$~\eqref{eq:power:u:IL} can be used efficiently for the URLLC allocation.

\section{Simplified solutions for power allocation}
\label{sec:algorithms}
Even if a closed-form for computing the power~\eqref{eq:outage:ufull} does not exist, we may state a non-increasing property of the outage probability, which is useful for the proposed simplified solutions.
\begin{proposition} \label{theo:non-increasing}
Assuming fixed values of $r_u$, and $\mb{P}_e$, the outage probability~\eqref{eq:outage:ufull} is a non-increasing monotone function of $\mb{P}_u$. % w.r.t. each direction given by $\mb{P}_u$.
\end{proposition}
\begin{proof}
See Appendix~\ref{proof:non-increasing}.
\end{proof}
% \begin{proof}
% The $\log$ function is monotonic non-decreasing w.r.t. each $P_u(f)$. Hence, reducing the power coefficient on a single mRB keeping the other power terms unchanged will reduce the mutual information $I_u$. 
% More formally, if we consider the vector $\mb{P}_u'$ obtained decreasing $P_u(f)$, $f\in\mc{F}_u$, by $\delta > 0$, while the other $P_u(g)$, $\forall g\in\mc{F}_u \setminus \{f\}$, and $P_e(j)$, $\forall j\in\mc{F}_u$, are kept the same, the mutual information results
% \begin{equation}
%     F_u I_u(\mb{P}_u') = \log_2\left( 1 + \frac{\gamma_u(f) (P_u(f) - \delta)}{1 + \gamma_u(f) P_e(f)} \right) + \sum_{g\in\mc{F}_u\setminus \{f\}} \log_2\left( 1 + \frac{\gamma_u(f)P_u(f)}{1 + \gamma_u(f) P_e(f)} \right) \le F_u I_u(\mb{P}_u).
% \end{equation}
% Therefore, the outage probability may only result larger (or equal), i.e.,
% \begin{equation}
%     p_u(\mb{P}_u') = \Pr\{I_u(\mb{P}_u') \le r_u\} \ge p_u(\mb{P}_u) = \Pr\{I_u(\mb{P}_u) \le r_u\}.
% \end{equation}
% The same relation is in fact extended for each $\mb{0} \preceq \mb{P}_u' \preceq \mb{P}_u$, which completes the proof.
% \end{proof}

\subsection{A lookup table estimation of URLLC outage probability}
\label{sec:feasible}
To estimate the outage probability, we will make use of a lookup table based on Monte Carlo simulations.
From \eqref{eq:outage:ufull}, the outage probability depends on the value of several parameters:  $r_u$, $\Gamma_u$, $\mc{F}_u$, $\mb{P}_e$ and $\mb{P}_u$. In theory, we can tabulate the outage probability with a specific entry for each of these parameters. However, populating the table for all the possible vectors $\mb{P}_e$ and $\mb{P}_u$ will require an extremely large number of trials and a huge dimension of the table itself. Furthermore, the complexity increases with the number of channels $\mc{F}_u$ employed, obtaining even bigger tables for greater frequency sets.
To overcome all these problems, we tabulate the outage probability assuming that the same values of $P_u$ and the same values of $P_e$ are used on all the considered mRBs. 
Let us denote as 
\begin{equation} \label{eq:montecarlopu}
    \hat{p}_u(P_u, P_e, \Gamma_u, \mc{F}_u, r_u ) = \lim_{n\rightarrow \infty} \frac{1}{n} \sum_{i=1}^n \mathbb{1}\left\{\sum_{f\in\mc{F}_u} \log\left( 1 + \frac{P_u \gamma_u(f)}{1 + P_e \gamma_u(f)}\right) \le F_u r_u \right\}
\end{equation}
the Monte Carlo estimation of the outage probability $p_u$ when $P_u(f) = P_u$ and $P_e(f) = P_e$ $ \forall f\in\mc{F}_u$, where $\mathbb{1}\{\cdot\}$ is the indicator function. \FS{The number of elements of the table is denoted as $L$ and depends on the quantity of parameters under test for power, frequency, SNR, and spectral efficiency values.}
From now on, we will recall eq.~\eqref{eq:montecarlopu} to address the whole lookup table for simplicity.

%In the following, we will propose an algorithm using the data from this table to find a feasible solution for the URLLC allocation process.
%In the remainder of the Section, we exploit the previous findings to propose two different algorithms able to solve both the OMA and NOMA allocation problems~\eqref{op:oma}-\eqref{op:noma}.

\subsection{A lookup table-based feasible algorithm for power allocation}
Both allocation problems~\eqref{op:oma} and \eqref{op:noma}
can be solved following a sequential algorithm, as summarized in Algorithm \ref{alg:N-fea}.
Given the channels $\mc{F}_e$, the optimal power allocation for user $e$, which transmits without any interference, is computed by means of \eqref{eq:power:e}.
% we load the portion of the table concerning these parameters $\hat{p}_u(\cdot, \cdot, \Gamma_u, \mc{F}_u, r_u)$. Then, we compute the eMBB power coefficients $P_e(f)$ through~\eqref{eq:power:e}.
Subsequently, given the known values $\Gamma_u$, $r_u$ and $\mc{F}_u$, and to compute the feasible power $P_u$, we consider the mRB experiencing the worst interference, denoted as
\begin{equation}
    f^\text{max} = \argmax_f P_e(f).
\end{equation}
Then, from the lookup table $\hat{p}_u(\cdot, {P}_e, \Gamma_u, \mc{F}_u, r_u)$, we obtain the minimum $P_u$ that meets the outage probability constraint when the interference is given by $P_e(f^\text{max})$, i.e.,
\begin{equation} \label{eq:fea:minP}
    P_u = \min \left\{ P \,|\, \hat{p}_u(P, P_e(f^\text{max}), \Gamma_u, \mc{F}_u, r_u ) \le \epsilon_u \right\}.
\end{equation}
In other words, we found the feasible power needed for the transmission assuming that all channels experience the strongest interference. Finally, we set:
\begin{equation}
    \label{eq:alg:fea}
    P_u^*(f) = \max\left\{P_u, P_u^\text{SIC}(f)\right\}, \quad \forall f\in\mc{F}_u.
\end{equation}
The procedure is summarized in Algorithm~\ref{alg:N-fea}. In the following, we show that this algorithm provides a feasible solution to the allocation process.

\begin{proposition} \label{theo:feasible}
Algorithm~\ref{alg:N-fea} guarantees a feasible solution of OMA allocation problem~\eqref{op:oma} and NOMA allocation problem~\eqref{op:noma}.
\end{proposition}
\begin{proof}
The allocation found with Algorithm~\ref{alg:N-fea} meets all the constraints of problems~\eqref{op:oma} and~\eqref{op:noma} and such is a \emph{feasible} solution. In facts constraints~\eqref{op:oma:pe} and~\eqref{op:noma:pe} are met by employing~\eqref{eq:power:e}. Constraints \eqref{op:oma:pupe} is satisfied by construction by virtue of orthogonal allocation. Constraints \eqref{op:oma:ou} and \eqref{op:noma:ou} are addressed by employing the lookup table as in \eqref{eq:fea:minP} and, finally, \eqref{op:noma:ue} is met by choosing $P_u^*(f)$ as in \eqref{eq:alg:fea} $\forall f\in\mc{F}_u$.
% In the case of OMA, we have $\mc{F}_u \cap \mc{F}_e = \emptyset$ (see~\eqref{eq:Fe}), so the orthogonality requirement is satisfied by definition. Moreover, the eMBB power coefficient obtained from eq.~\eqref{eq:power:e} guarantees that $p_e = 0$~\cite{Tse2005}.
%The remaining of the proof is given by~\cite[Proposition 1]{saggese2021noma}.
% Please refer to Appendix~\ref{proof:feasible}.
\end{proof}
It is worth noting that this scheme is generally sub-optimal in terms of power spent, considering that we constrain all channels to act as the worst one. 
On the other hand, for OMA allocation, we can demonstrate the following proposition.
\begin{proposition} \label{theo:oma-optimality}
Given $\Gamma_u$, $\mc{F}_u$, and $r_u$, if the lookup table~\eqref{eq:montecarlopu} contains an entry where $\hat{p}_u(P_u, 0, \Gamma_u, \mc{F}_u, r_u) = \epsilon_u$ exactly, Algorithm~\ref{alg:N-fea} provide the optimal solution of the URLLC OMA allocation problem~\eqref{eq:op:u}, and the optimal power coefficient is $P_u^*(f) = P_u$, $\forall f \in\mc{F}_u$.
\end{proposition}
\begin{proof}
See Appendix~\ref{proof:oma-optimality}.
\end{proof}
% \begin{proof}
% For the OMA allocation, the power coefficient of $e$ user is zero for every resources given to $u$, i.e. $P_e(f) = 0$, $\forall f\in\mc{F}_u$. %Hence, every channel experiences the same channel condition. 
% In these conditions of i.i.d. parallel channels without interference, the minimum outage probability is reached when the same power coefficient is allocated to each channel~\cite{Tse2005}.
% As a matter of fact, Algorithm~\ref{alg:N-fea} outputs the power coefficient $P_u$ given by~\eqref{eq:fea:minP} for each mRB, due to the lack of SIC constraint for the OMA case. The resulting vector $\mb{P}_u^* = [P_u, \dots, P_u]\T$ is at least a feasible solution of problem~\eqref{op:oma}, due to Proposition~\ref{theo:feasible}.
% The resulting mutual information is
% \[
% I_u(\mb{P}_u^*, \mb{P}_e) = I_u(P_u, 0) = \frac{1}{F_u} \sum_{f\in\mc{F}_u} \log_2\left(1 + \gamma_u(f) P_u \right).
% \]
% Proposition~\ref{theo:non-increasing} is a generalization of this particular case; thus, the outage probability $p_u(\mb{P}_u^*)$ is monotonically non-increasing in $P_u$. %Moreover, the objective function of~\eqref{eq:op:u} is monotonically non-decreasing
% Having obtained $P_u$ following~\eqref{eq:fea:minP}, there is no other power coefficient $P_u' \le P_u$ such as $\hat{p}_u(P_u', 0, \Gamma_u, \mc{F}_u, r_u) \le \epsilon_u$. 
% Eventually, if the table has been populated with enough different values of $P_u$, we will find that $p_u(P_u, 0, \Gamma_u, \mc{F}_u, r_u) = \epsilon_u$, and the optimal solution is reached.
% \end{proof}
%
In practical terms, Proposition~\ref{theo:oma-optimality} implies that the solution given by Algorithm~\ref{alg:N-fea} is reasonably close to the optimal one for the OMA allocation. The larger is the set of $P_u$ values in the lookup table~\eqref{eq:montecarlopu}, the closer  we get to the optimal results.

\FS{The computational complexity of Algorithm~\ref{alg:N-fea} depends on: the complexity of the water-filling algorithm, being $\mc{O}(F)$ for the single-user case~\cite[Proposition 1]{Palomar2005}; the complexity of the maximum operator, whose worst-case is $\mc{O}(F)$; the complexity of retrieving the minimum power from the lookup table, which is linear with the number of the element of the table involved in the power coefficient retrieval $\mc{O}(L^\text{fea})$. The total computational complexity is $\mc{O}(3F + L^\text{fea})$, where $L^\text{fea}$ is generally the dominant term. It is worth noting that $L^\text{fea} \ll L$, and the former is equal to the number of $P_u$ values tested in the lookup table, considering that the scheduler is previously informed of $\Gamma_u, \mc{F}_u$, and $r_u$.
We remark that the initialization of the procedure is not considered because the computational load required for creating~\eqref{eq:montecarlopu} can be performed offline. In other words, the population of the look-up table does not directly increase the complexity of the resolution of the URLLC allocation problem.}

\begin{algorithm}
\footnotesize
\caption{Feasible algorithm (N-fea)}
\label{alg:N-fea}
\textbf{Initialize:} Populate the table $\hat{p}_u(P_u, P_e, \Gamma_u, \mc{F}_u, r_u )$\; 
Compute $\mb{P}_e$ through~\eqref{eq:power:e}\;
Compute $\mb{P}_u^\text{SIC}$ through~\eqref{eq:power:sic}\;
$f^{\max} = \argmax_f P_e(f)$\;
$P_u = \min \{P_u \,|\, \hat{p}_u(P_u, P_e(f^{\max}), \Gamma_u, \mc{F}_u, r_u ) \le \epsilon_u\}$\;
$P_u^*(f) = \max\left\{P_u, P_u^\text{SIC}(f)\right\}$, $\forall f \in\mc{F}_u$\;
\textbf{Output:} $\mb{P}_u^*$
\end{algorithm}
% \vspace{-0.5cm}

\subsection{Constrained block coordinated descent (BCD) optimization}
\label{sec:exhaustive}
To overcome the sub-optimality given by Algorithm~\ref{alg:N-fea}, we propose an iterative algorithm based on BCD optimization, a class of algorithms able to effectively solve large-scale optimization problems~\cite{Xu2017}.
The main idea is based on non-increasing property of the outage probability~\eqref{eq:outage:ufull} given in Proposition~\ref{theo:non-increasing}.

We start from the solution given by Algorithm~\ref{alg:N-fea}. For each iteration, we update the power coefficients one after another, following a deterministic order. Without loss of generality, we assume that the predetermined order is the natural order of set $\mc{F}_u$, i.e., $0, 1, 2, \dots, F_u - 1$.
To take into account the dimension-wise updating process, we denote the power coefficient vector at iteration $i$ where the first $f\in\mc{F}_u$ terms have been updated as 
\begin{equation}
    \mb{P}_u^{(i, f)} =  [P_u^{(i)}(0), \dots, P_u^{(i)}(f-1), P_u^{(i-1)}(f), \dots, P_u^{(i+1)}(F_u-1)]\T.
\end{equation}
The aim is to update each power coefficient only if the resulting outage probability satisfies the reliability requirement. To guarantee the feasibility of using $\mb{P}_u^{(i, f)}$, the power coefficient of frequency component $f\in\mc{F}_u$ is updated as follows
\begin{equation}
    \label{eq:pupdate}
    P_u^{(i)}(f)= 
    \begin{cases}
        \displaystyle
        \max\left\{P_u^\text{SIC}(f), P_u^{(i-1)}(f) - \mu^{(i)}\right\}, \quad \text{if } \hat{p}_u(\mb{P}_u^{(i, f)}, \mb{P}_e, \Gamma_u, \mc{F}_u, r_u ) \le \epsilon_u,  \\
        P_u^{(i-1)}(f), \quad \text{otherwise},
    \end{cases}
\end{equation}
where $\mu^{(i)} > 0$ is the step size at iteration $i$, $\hat{p}_u(\mb{P}_u^{(i, f)}, \mb{P}_e, \Gamma_u, \mc{F}_u, r_u )$ is the Monte Carlo estimation of the outage probability setting as input vector $\mb{P}_u^{(i,f)}$.
To avoid unnecessary computations, when a power coefficient reaches $P_u^\text{SIC}(f)$, the updating process for that frequency is not performed anymore.
To assure the convergence near to the optimal solution, we halve the value of step size when the updating rule~\eqref{eq:pupdate} lets the vector unchanged.
Hence, the step size updating rule is
\begin{equation}
    \label{eq:mupdate}
    \mu^{(i)} =
    \begin{cases}
        \displaystyle
        \mu^{(i-1)} / 2, \quad \text{if } \mb{P}_u^{(i, F_u)} = \mb{P}_u^{(i-1, F_u)} \\
        \mu^{(i-1)}, \quad \text{otherwise}.
    \end{cases}
\end{equation}
The algorithm stops when the step size is lower or equal to a certain threshold $\tau$. The overall algorithm is summarized in Algorithm~\ref{alg:ex}.

Taking into account that the objective function of~\eqref{eq:op:u} is monotonically non-decreasing, we can prove the following Proposition.
\begin{proposition} \label{theo:convergence}
There exists a minimum value of the threshold $\tau$ for which Algorithm~\ref{alg:ex} provides the optimal solution for problem~\eqref{eq:op:u}.
\end{proposition}
\begin{proof}
See Appendix~\ref{proof:convergence}.
\end{proof}
\FS{In theory, the optimal solution is guaranteed to be obtained only in the asymptotic behavior $\tau \rightarrow 0$. When the threshold tends to zero, the algorithm iterates until the border of the feasible set is found (see Appendix~\ref{proof:convergence}). However, this leads to a number of iterations of the algorithm that tends to the infinite; hence, the algorithm cannot be implemented in this way.}
In practice, we can set $\tau$ as the minimum variation of power allowed at the BS to obtain the closest to the optimum solution.

\FS{The complexity of each iteration $i$ of the BCD algorithm is dominated by the estimation of outage probability when applying update rule~\eqref{eq:pupdate}, %. This arises a complexity linear with the number of samples generated, 
i.e., $\mc{O}(n F_u)$ (see eq.~\eqref{eq:montecarlopu}). Cycling through all the frequencies, the complexity for the iteration $i$ is in the worst case $\mc{O}(n F_u^2)$. Denoting the number of iterations as a function of the threshold as $I(\tau)$, we have a total complexity of $\mc{O}(n F_u^2 I(\tau))$. Even setting a $\tau$ to provide a low number of iterations, the dominant term remains $n$, which needed to be a high value to obtain an accurate estimation of the outage probability. Therefore, this algorithm can be hardly used in a real scenario, but it represents a good benchmark for evaluating the performance of the NOMA approach.}

\begin{algorithm}
\footnotesize
\caption{Constrained BCD (N-BCD)}
\label{alg:ex}
\textbf{Initialize:}
$\mu^{(0)} > 0$, Compute $\mb{P}_e$ from~\eqref{eq:power:e}, $\mb{P}_u^{(0, F_u)}$ using Algorithm~\ref{alg:N-fea}; $i \leftarrow 1$, $\mc{F}^+ \leftarrow \mc{F}_u$\;
\While{$\mu^{(i)} > \tau$}
{
	\For {$f \in\mc{F}^+$}
	{
		Update $P_u^{(i)}(f)$ using~\eqref{eq:pupdate}\;
		\If{$P_u^{(i)}(f) \ge P_u^\text{SIC}(f)$}{$\mc{F}^+ \leftarrow \mc{F}^+ \setminus \{f\}$}
	}
	$i \leftarrow i + 1$\;
	Update $\mu^{(i)}$ using~\eqref{eq:mupdate}
}
\textbf{Output:} $\mb{P}_u^* \leftarrow \mb{P}_u^{(i-1, F_u)}$
\end{algorithm}
\normalsize

\section{Numerical results}
\label{sec:results}
In this Section, the performance comparison between OMA and NOMA for spectrum slicing of eMBB and URLLC traffic is presented. \mm{In the absence of other literature on the subject, our focus is to investigate in depth the various effects that play a role in the computation of the power consumption, such as channel diversity, interference cancellation, and URLLC requirements.}
We consider a resource grid formed by $F = 12$ frequencies and $M = 7$ mini-slots. We consider that the time of a slot is $T = 1$ ms, obtaining $T_m = T/M = 1/7$ ms as the minimum URLLC time slot of the 5G NR standard~\cite{3gpp:access}.
Finally, each mRB has bandwidth $\Delta_f = 180$ kHz.

The reliability requirement of URLLC user is set as $\epsilon_u = 10^{-5}$. The fading channels for both users are Rayleigh distributed with scale parameter $\sqrt{\Gamma_i}$, $i \in \{e, u\}$.
In the Monte Carlo simulation used to obtain the lookup table~\eqref{eq:montecarlopu}, we vary the powers $P_i$, $i \in \{e, u\}$ of each mRB from -30 dBm to 30 dBm per mRB, with granularity 1 dBm.

For each instance of simulation, we place the users in a cell of 500 m in radius and we compute the power consumption using both OMA and NOMA. The distance $d_i$, $i \in \{e, u\}$ of the users from the BS is computed by inverting the well-known wireless path loss formulation, i.e., $ d_i = \sqrt[\psi]{10^{G/10} L_i \left(\frac{c}{f_0 4 \pi}\right)^2 d_0^{\psi-2}}$,
% \begin{equation*}
%     d_i = \sqrt[\psi]{10^{G/10} L_i \left(\frac{c}{f_0 4 \pi}\right)^2 d_0^{\psi-2}}
% \end{equation*}
where $G = 17.15$ dB is the overall antenna gain of BS and user; $f_0 = 2$ GHz is the central working frequency; $c$ is the speed of light; $d_0 = 10$ m is the free-space region of the cell near to the BS; $\psi = 4$ is the path loss exponent; $L_i = \sigma^2 / \Gamma_i$ is the path loss. We set the receiver noise as $\sigma^2 = -108$ dBm.

In the case of NOMA, we show the results of Algorithms~\ref{alg:N-fea} and~\ref{alg:ex}, labeled as N-fea and N-BCD, respectively. For this paradigm, we always present the results obtained setting $F_u = F = 12$. We remark that the choice of this parameter does not influence the eMBB allocation while increasing the number of resources available reduces the power coefficients needed for SIC (see~\eqref{eq:power:sic}).
We set the threshold of the BCD approach as $\tau = 10^{-7}$.
In the case of OMA, we present the results for different frequencies reserved for the URLLC data stream transmission. In particular, we set $F_u \in \{3,6,9\}$, labeled as O-3, O-6, O-9, respectively. 
To correctly compare the performance of schemes with different number of resources available, we fix the number of bits to be transmitted in the whole slot as $N_e = 8640$ and $N_u = \frac{2160}{7}$. In this way, the spectral efficiency for NOMA with $M_u = 1$ are $r_u = 1$ bit/s/Hz and $r_e = 4$ bit/s/Hz. For the other $M_u$ values, the spectral efficiencies can be computed trough~\eqref{eq:bit2se}.

\begin{figure}[htb]
    \centering
    % This file was created with tikzplotlib v0.9.12.
\begin{tikzpicture}

\definecolor{color0}{rgb}{1,0.270588235294118,0}
\definecolor{color1}{rgb}{1,0.549019607843137,0}

\begin{axis}[
legend cell align={left},
legend style={
  fill opacity=0.8,
  draw opacity=1,
  text opacity=1,
  at={(0.98,0.02)},
  anchor=south east,
  draw=white!80!black
},
tick align=outside,
tick pos=left,
x grid style={white!69.0196078431373!black},
xlabel={\(\displaystyle d_e\) [m]},
xmajorgrids,
xmin=0, xmax=500,
xtick style={color=black},
y grid style={white!69.0196078431373!black},
ylabel={\(\displaystyle P_{u}^\mathrm{SIC}\) [dBm]},
ymajorgrids,
ymin=-35.7136722352853, ymax=39.9976016195109,
ytick style={color=black}
]
\addplot [semithick, blue,  mark=asterisk,  mark options={solid,fill opacity=0}]
table {%
619.498731755915 24.1324802553991
464.558511632561 19.0617905292523
348.369737769362 13.907299415622
261.24044905991 9.2062880871926
195.902699993437 4.04872547884687
146.906300317672 -0.7528385547802
110.164183922677 -5.84722705522243
82.611483599451 -11.1012656896391
61.9498731755915 -16.2543266956548
46.4558511632561 -21.1761558253804
34.8369737769362 -26.7498942800637
26.124044905991 -32.272250696431
};
\addlegendentry{$r_e = 2$}
\addplot [semithick, red,  mark=triangle,  mark options={solid,rotate=180,fill opacity=0}]
table {%
619.498731755915 30.4666579005981
464.558511632561 25.4508262857302
348.369737769362 20.2078496146187
261.24044905991 15.5833254985537
195.902699993437 10.2201754906476
146.906300317672 4.98072444921688
110.164183922677 -1.01289264907221
82.611483599451 -7.85655938961304
61.9498731755915 -14.0272913863048
46.4558511632561 -19.6711867363474
34.8369737769362 -26.2931184262954
26.124044905991 -32.1022561038476
};
\addlegendentry{$r_e = 4$}
\addplot [semithick, green,  mark=square,  mark options={solid,fill opacity=0}]
table {%
619.498731755915 36.5561800806565
464.558511632561 31.5425175372835
348.369737769362 26.281672985994
261.24044905991 21.6460268248102
195.902699993437 15.7577125313655
146.906300317672 8.91040843710603
110.164183922677 1.09979794872688
82.611483599451 -7.66215038136391
61.9498731755915 -13.9381546578608
46.4558511632561 -19.6711811062957
34.8369737769362 -26.2931184262918
26.124044905991 -32.1022561039078
};
\addlegendentry{$r_e = 6$}
\end{axis}

\end{tikzpicture}
    \caption{Average power consumption required by SIC as a function of $d_e$, $F_u = 12$.}
    \label{fig:sic}
    \vspace{-0.3cm}
\end{figure}
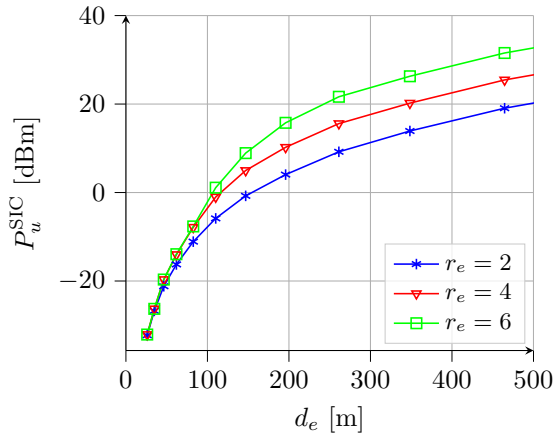

Firstly, we study the behavior of the power needed for the SIC process given in eq.~\eqref{eq:power:sic}.
Fig.~\ref{fig:sic} shows the overall power required for the SIC process $P_u^\text{SIC} = \sum_{f\in\mc{F}_u} P_u^\text{SIC}(f)$ as a function of the distance of the eMBB user $d_e$.
As expected, the closer the eMBB user to the BS the smaller the power needed for the SIC. We can also see that the gap between different spectral efficiencies starts to be relevant when $d_e \ge 100$ m. For higher distances, increasing the spectral efficiency increases the value of eMBB power $\mb{P}_e$ given by solution~\eqref{eq:power:e}; for lower distances, the evaluated power is dominated by the terms $1 / \gamma_e(f)$. We remark that a low $d_e$ means a high average SNR $\Gamma_e$ leading to a higher probability of having high values of $\gamma_e(f)$.

\begin{figure}[bht]
    \centering
    % This file was created with tikzplotlib v0.9.12.
\begin{tikzpicture}

\definecolor{color0}{rgb}{1,0.270588235294118,0}
\definecolor{color1}{rgb}{1,0.549019607843137,0}
\definecolor{color2}{rgb}{1,0.843137254901961,0}

\begin{axis}[
legend cell align={left},
legend style={fill opacity=0.8, draw opacity=1, text opacity=1, draw=white!80!black,
at={(0.01, 0.01)}, anchor=south west},
log basis y={10},
tick align=outside,
tick pos=left,
x grid style={white!69.0196078431373!black},
xlabel={\(\displaystyle P_u(f)\) [dBm]},
xmajorgrids,
xmin=-15, xmax=15,
xtick style={color=black},
xtick={-30,-20,-10,0,10,20,30},
xticklabels={
  \(\displaystyle {\ensuremath{-}30}\),
  \(\displaystyle {\ensuremath{-}20}\),
  \(\displaystyle {\ensuremath{-}10}\),
  \(\displaystyle {0}\),
  \(\displaystyle {10}\),
  \(\displaystyle {20}\),
  \(\displaystyle {30}\)
},
y grid style={white!69.0196078431373!black},
ylabel={\(\displaystyle p_u\)},
ymajorgrids,
ymin=1e-06, ymax=5,
ymode=log,
ytick style={color=black},
ytick={1e-07,1e-06,1e-05,0.0001,0.001,0.01,0.1,1,10,100},
yticklabels={
  \(\displaystyle {10^{-7}}\),
  \(\displaystyle {10^{-6}}\),
  \(\displaystyle {10^{-5}}\),
  \(\displaystyle {10^{-4}}\),
  \(\displaystyle {10^{-3}}\),
  \(\displaystyle {10^{-2}}\),
  \(\displaystyle {10^{-1}}\),
  \(\displaystyle {10^{0}}\),
  \(\displaystyle {10^{1}}\),
  \(\displaystyle {10^{2}}\)
}
]
\addplot [semithick, red!54.5098039215686!black, mark=asterisk,   mark options={solid,fill opacity=0}]
table {%
-30 1
-29 1
-28 1
-27 1
-26 1
-25 1
-24 1
-23 1
-22 1
-21 1
-20 1
-19 1
-18 1
-17 1
-16 1
-15 1
-14 1
-13 1
-12 1
-11 1
-10 1
-9 1
-8 1
-7 1
-6 1
-5 1
-4 1
-3 1
-2 1
-1 1
0 1
1 1
2 0.9999536
3 0.9691144
4 0.6467767
5 0.2069503
6 0.0348698
7 0.0037927
8 0.0003149
9 2.38e-05
10 1.4e-06
11 2e-07
12 0
13 0
14 0
15 0
16 0
17 0
18 0
19 0
20 0
21 0
22 0
23 0
24 0
25 0
26 0
27 0
28 0
29 0
30 0
};
\addlegendentry{$M_u=1$}
\addplot [semithick, color0, mark=triangle,   mark options={solid,rotate=180,fill opacity=0}]
table {%
-30 1
-29 1
-28 1
-27 1
-26 1
-25 1
-24 1
-23 1
-22 1
-21 1
-20 1
-19 1
-18 1
-17 1
-16 1
-15 1
-14 1
-13 1
-12 1
-11 1
-10 1
-9 1
-8 1
-7 1
-6 1
-5 1
-4 1
-3 1
-2 0.9999923
-1 0.979494
0 0.6627628
1 0.1945952
2 0.0275155
3 0.0024276
4 0.0001673
5 1.03e-05
6 4e-07
7 0
8 0
9 0
10 0
11 0
12 0
13 0
14 0
15 0
16 0
17 0
18 0
19 0
20 0
21 0
22 0
23 0
24 0
25 0
26 0
27 0
28 0
29 0
30 0
};
\addlegendentry{$M_u=2$}
\addplot [semithick, color1, mark=square,   mark options={solid,fill opacity=0}]
table {%
-30 1
-29 1
-28 1
-27 1
-26 1
-25 1
-24 1
-23 1
-22 1
-21 1
-20 1
-19 1
-18 1
-17 1
-16 1
-15 1
-14 1
-13 1
-12 1
-11 1
-10 1
-9 1
-8 1
-7 1
-6 1
-5 0.9993661
-4 0.8926309
-3 0.3984917
-2 0.0719733
-1 0.0068783
0 0.00046
1 2.62e-05
2 1.6e-06
3 0
4 0
5 0
6 0
7 0
8 0
9 0
10 0
11 0
12 0
13 0
14 0
15 0
16 0
17 0
18 0
19 0
20 0
21 0
22 0
23 0
24 0
25 0
26 0
27 0
28 0
29 0
30 0
};
\addlegendentry{$M_u=4$}
\addplot [semithick, color2, mark=triangle,   mark options={solid,fill opacity=0}]
table {%
-30 1
-29 1
-28 1
-27 1
-26 1
-25 1
-24 1
-23 1
-22 1
-21 1
-20 1
-19 1
-18 1
-17 1
-16 1
-15 1
-14 1
-13 1
-12 1
-11 1
-10 1
-9 1
-8 0.9999904
-7 0.9730572
-6 0.5982906
-5 0.1420089
-4 0.0156767
-3 0.0011082
-2 6.17e-05
-1 1.9e-06
0 1e-07
1 1e-07
2 0
3 0
4 0
5 0
6 0
7 0
8 0
9 0
10 0
11 0
12 0
13 0
14 0
15 0
16 0
17 0
18 0
19 0
20 0
21 0
22 0
23 0
24 0
25 0
26 0
27 0
28 0
29 0
30 0
};
\addlegendentry{$M_u=7$}
\addplot [semithick, red!54.5098039215686!black, dashed, mark=asterisk,   mark options={solid,fill opacity=0}]
table {%
-30 1
-29 1
-28 1
-27 1
-26 1
-25 1
-24 1
-23 1
-22 1
-21 1
-20 1
-19 1
-18 1
-17 1
-16 1
-15 1
-14 1
-13 1
-12 1
-11 1
-10 1
-9 1
-8 1
-7 1
-6 1
-5 1
-4 0.9999982
-3 0.9998083
-2 0.9948538
-1 0.9509226
0 0.7904765
1 0.508101
2 0.2382616
3 0.0812378
4 0.020686
5 0.0041123
6 0.0006788
7 9.2e-05
8 9.3e-06
9 1e-06
10 3e-07
11 0
12 0
13 0
14 0
15 0
16 0
17 0
18 0
19 0
20 0
21 0
22 0
23 0
24 0
25 0
26 0
27 0
28 0
29 0
30 0
};
% \addlegendentry{$M_u=1$}
\addplot [semithick, color0, dashed, mark=triangle,   mark options={solid,rotate=180,fill opacity=0}]
table {%
-30 1
-29 1
-28 1
-27 1
-26 1
-25 1
-24 1
-23 1
-22 1
-21 1
-20 1
-19 1
-18 1
-17 1
-16 1
-15 1
-14 1
-13 1
-12 1
-11 1
-10 1
-9 1
-8 0.9999969
-7 0.9997144
-6 0.9926545
-5 0.9348382
-4 0.7439817
-3 0.4449298
-2 0.1901533
-1 0.0586328
0 0.0135658
1 0.0024472
2 0.0003625
3 4.57e-05
4 3.7e-06
5 6e-07
6 0
7 0
8 0
9 0
10 0
11 0
12 0
13 0
14 0
15 0
16 0
17 0
18 0
19 0
20 0
21 0
22 0
23 0
24 0
25 0
26 0
27 0
28 0
29 0
30 0
};
% \addlegendentry{$M_u=2$}
\addplot [semithick, color1, dashed, mark=square,   mark options={solid,fill opacity=0}]
table {%
-30 1
-29 1
-28 1
-27 1
-26 1
-25 1
-24 1
-23 1
-22 1
-21 1
-20 1
-19 1
-18 1
-17 1
-16 1
-15 1
-14 1
-13 1
-12 0.9999998
-11 0.9999463
-10 0.9976715
-9 0.9684894
-8 0.8331347
-7 0.5557753
-6 0.2659299
-5 0.0906304
-4 0.0226914
-3 0.0043944
-2 0.0006788
-1 9.08e-05
0 1.07e-05
1 7e-07
2 0
3 0
4 0
5 0
6 0
7 0
8 0
9 0
10 0
11 0
12 0
13 0
14 0
15 0
16 0
17 0
18 0
19 0
20 0
21 0
22 0
23 0
24 0
25 0
26 0
27 0
28 0
29 0
30 0
};
% \addlegendentry{$M_u=4$}
\addplot [semithick, color2, dashed, mark=triangle,   mark options={solid,fill opacity=0}]
table {%
-30 1
-29 1
-28 1
-27 1
-26 1
-25 1
-24 1
-23 1
-22 1
-21 1
-20 1
-19 1
-18 1
-17 1
-16 1
-15 1
-14 0.9999895
-13 0.9991509
-12 0.9838534
-11 0.8890443
-10 0.6457468
-9 0.3404704
-8 0.1272906
-7 0.0344958
-6 0.0071385
-5 0.001154
-4 0.000156
-3 1.69e-05
-2 1.7e-06
-1 2e-07
0 0
1 0
2 0
3 0
4 0
5 0
6 0
7 0
8 0
9 0
10 0
11 0
12 0
13 0
14 0
15 0
16 0
17 0
18 0
19 0
20 0
21 0
22 0
23 0
24 0
25 0
26 0
27 0
28 0
29 0
30 0
};
% \addlegendentry{$M_u=7$}
\end{axis}

\end{tikzpicture}
    \caption{Outage probability $p_u$ versus $P_u(f)$, for $\Gamma_u = 30$ dB, $F_u = 12$ and  different $M_u$. $P_u(f) = P_u$, $P_e(f) = P_e$, $\forall f \in\mc{F}_u$. The solid lines are for $P_e(f) = 0$ dBm while dashed lines represent $P_e(f) = -\infty$ dBm, i.e., no interference.}
    \label{fig:outages_M}
    \vspace{-0.3cm}
\end{figure}
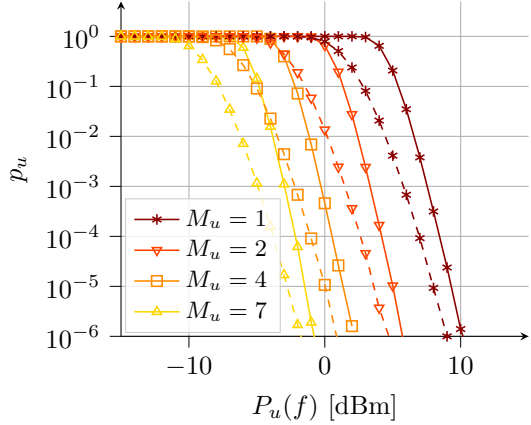

In the following, we show the effect of the URLLC transmission time $M_u$ on the power allocation.
Fig.~\ref{fig:outages_M} shows the outage probability $p_u$ as a function of the power $P_u(f)$, when all transmitting and interference power coefficients are the same for each mRB, i.e., $P_u(f) = P_u$, $P_e(f) = P_e$, $\forall f \in\mc{F}_u$, $\Gamma_u = 30$ dB, and $F_u = 12$. The dashed lines are the outage probability curves with no interference, i.e. $P_e = -\infty$ dBm, while the solid lines represent the outage probability with $P_e = 0$ dBm. We present different plots for different $M_u$. Increasing the value of $M_u$, will reduce the spectral efficiency $r_u$ leading to a reduction of the power needed to meet the target outage probability $\epsilon_u$. Therefore, the best solution is design a system able to exploit the whole latency requirement.
Having shown this result, we will focus only on the results of $M_u = 1$ in the remainder of the Section, remarking that increasing $M_u$ reduces the overall power spent for all the schemes presented.

% The following results are made fixing the position of the eMBB user. In particular, the positioning of the two users is made as follows: the URLLC user $u$ is placed to have the average channel gain power equal to $\Gamma_u$ dB; the eMBB user is then randomly positioned following a uniform distribution in the cell. 
% The path loss is computed assuming no shadowing and setting the global antenna gain to 17.5 dB, the path loss exponent to 4 and the receiver noise as $-92$ dBm. Finally, the results are averaged for all the different positions of the eMBB user.

\begin{figure}[hbt]
    \centering
    \begin{subfigure}[c]{.45\textwidth}
        \centering
        % This file was created with tikzplotlib v0.9.12.
\begin{tikzpicture}

\definecolor{color0}{rgb}{1,0.270588235294118,0}
\definecolor{color1}{rgb}{1,0.549019607843137,0}

\begin{axis}[
legend cell align={left},
legend style={
  fill opacity=0.8,
  draw opacity=1,
  text opacity=1,
  at={(0.01,0.6)},
  anchor=west,
  draw=white!80!black
},
tick align=outside,
tick pos=left,
x grid style={white!69.0196078431373!black},
xlabel={\(\displaystyle d_u\) [m]},
xmajorgrids,
xmin=0, xmax=500,
xtick style={color=black},
y grid style={white!69.0196078431373!black},
ylabel={\(\displaystyle P^\mathrm{tot}\) [dBm]},
ymajorgrids,
ymin=13.2182961563393, ymax=40.063882982398,
ytick style={color=black}
]
\addplot [semithick, blue, mark=x, mark options={solid,fill opacity=0}]
table {%
10 15.1715931720987
35.7894736842105 15.1715932958874
61.5789473684211 15.171590754601
87.3684210526316 15.1713070984224
113.157894736842 15.1726579795975
138.947368421053 15.1998767953882
164.736842105263 15.3098236793924
190.526315789474 15.496928041253
216.315789473684 15.7224064281116
242.105263157895 15.9908177481826
267.894736842105 16.3542195018248
293.684210526316 16.7976381991224
319.473684210526 17.2968656613732
345.263157894737 17.8329527625142
371.052631578947 18.3880632091122
396.842105263158 18.948660842975
422.631578947368 19.5223650735725
448.421052631579 20.1521374390391
474.210526315789 20.8121683981057
500 21.4792678500731
};
\addlegendentry{N-fea}
\addplot [semithick, green!50!black, mark=o, mark options={solid,fill opacity=0}]
table {%
10 14.6971136351145
35.7894736842105 14.697131690914
61.5789473684211 14.7352165650249
87.3684210526316 14.8779375356049
113.157894736842 15.0224403221579
138.947368421053 15.134910529958
164.736842105263 15.2702172266142
190.526315789474 15.4438351684967
216.315789473684 15.6631809627694
242.105263157895 15.9281710322887
267.894736842105 16.2620818262119
293.684210526316 16.6534282104904
319.473684210526 17.0947372455849
345.263157894737 17.6052745111627
371.052631578947 18.1626784976198
396.842105263158 18.7444859510718
422.631578947368 19.3456883737941
448.421052631579 19.9877829894858
474.210526315789 20.6479820278921
500 21.3074982995075
};
\addlegendentry{N-BCD}
\addplot [semithick, red!54.5098039215686!black, mark=asterisk, mark options={solid,fill opacity=0}]
table {%
10 14.4385501029783
35.7894736842105 14.4752954203053
61.5789473684211 14.7525113922093
87.3684210526316 15.5797735985234
113.157894736842 17.1147128440292
138.947368421053 19.0748870106089
164.736842105263 21.1458872295994
190.526315789474 23.1818915977169
216.315789473684 25.215257779011
242.105263157895 27.0675090496013
267.894736842105 28.8478611180552
293.684210526316 30.4550098918348
319.473684210526 31.8736644419865
345.263157894737 33.125956970355
371.052631578947 34.2390561071284
396.842105263158 35.2369527726255
422.631578947368 36.1391690055815
448.421052631579 36.961240077524
474.210526315789 37.7155064480234
500 38.4118313103989
};
\addlegendentry{O-3}
\addplot [semithick, color0, mark=triangle, mark options={solid,rotate=180,fill opacity=0}]
table {%
10 19.2260637226592
35.7894736842105 19.2264379254716
61.5789473684211 19.2317084433669
87.3684210526316 19.2497712146487
113.157894736842 19.2926813251704
138.947368421053 19.3751006949347
164.736842105263 19.5175311454024
190.526315789474 19.7317833588177
216.315789473684 20.0371253648775
242.105263157895 20.4280460377359
267.894736842105 20.9386056338373
293.684210526316 21.535204635919
319.473684210526 22.1921531713582
345.263157894737 22.9235076587808
371.052631578947 23.6888648001822
396.842105263158 24.4560214238922
422.631578947368 25.2189646391627
448.421052631579 26.0025041187594
474.210526315789 26.781774955233
500 27.5392540887538
};
\addlegendentry{O-6}
\addplot [semithick, color1, mark=square, mark options={solid,fill opacity=0}]
table {%
10 38.7398916370567
35.7894736842105 38.7398916851497
61.5789473684211 38.7399103802604
87.3684210526316 38.739985245071
113.157894736842 38.7401664078381
138.947368421053 38.7405178984736
164.736842105263 38.7411334417492
190.526315789474 38.742107813843
216.315789473684 38.7436254278053
242.105263157895 38.7457390696843
267.894736842105 38.7486895564457
293.684210526316 38.7525115680445
319.473684210526 38.7573260923097
345.263157894737 38.7636970578647
371.052631578947 38.7716882469299
396.842105263158 38.7812907150035
422.631578947368 38.7927370370745
448.421052631579 38.8069182935457
474.210526315789 38.8238932312197
500 38.843629035759
};
\addlegendentry{O-9}
\end{axis}

\end{tikzpicture}
         \caption{Total power spent.}
        \label{fig:opt_power_ge50}
    \end{subfigure}
    ~
    \begin{subfigure}[c]{.45\textwidth}
        \centering
        % This file was created with tikzplotlib v0.9.12.
\begin{tikzpicture}

\definecolor{color0}{rgb}{1,0.270588235294118,0}
\definecolor{color1}{rgb}{1,0.549019607843137,0}

\begin{axis}[
legend cell align={left},
legend style={
  fill opacity=0.8,
  draw opacity=1,
  text opacity=1,
  at={(0.97,0.03)},
  anchor=south east,
  draw=white!80!black
},
log basis y={10},
tick align=outside,
tick pos=left,
x grid style={white!69.0196078431373!black},
xlabel={\(\displaystyle d_u\) [m]},
xmajorgrids,
xmin=0, xmax=500,
xtick style={color=black},
y grid style={white!69.0196078431373!black},
ylabel={\(\displaystyle p_u\)},
ymajorgrids,
ymin=3.05258614423389e-08, ymax=1.31758219583014e-05,
ymode=log,
ytick style={color=black}
]
\addplot [very thin, black]
table {%
0 1e-05
500 1e-05
};
\addlegendentry{$\epsilon_u$}
\addplot [semithick, blue, mark=x, mark options={solid,fill opacity=0}]
table {%
10 0
35.7894736842105 1e-9
61.5789473684211 5e-09
87.3684210526316 4.02203315488035e-08
113.157894736842 5.36056696431789e-07
138.947368421053 3.35350319993593e-06
164.736842105263 4.52868889144287e-06
190.526315789474 4.97082664111618e-06
216.315789473684 5.23358357558164e-06
242.105263157895 5.44909532541638e-06
267.894736842105 5.50177291041183e-06
293.684210526316 5.58620173349823e-06
319.473684210526 5.67063055658463e-06
345.263157894737 5.725505937967103e-06
371.052631578947 5.81468607389507e-06
396.842105263158 5.90787289081267e-06
422.631578947368 6.20105970773027e-06
448.421052631579 6.89424652464786e-06
474.210526315789 7.32145901573567e-06
500 7.30398198662931e-06
};
\addlegendentry{N-fea}
\addplot [semithick, green!50!black, mark=o, mark options={solid,fill opacity=0}]
table {%
10 4.83469431151243e-07
35.7894736842105 6.8853982038986e-06
61.5789473684211 8.69667795535271e-06
87.3684210526316 1e-05
113.157894736842 1e-05
138.947368421053 1e-05
164.736842105263 1e-05
190.526315789474 1e-05
216.315789473684 1e-05
242.105263157895 1e-05
267.894736842105 1e-05
293.684210526316 1e-05
319.473684210526 1e-05
345.263157894737 1e-05
371.052631578947 1e-05
396.842105263158 1e-05
422.631578947368 1e-05
448.421052631579 1e-05
474.210526315789 1e-05
500 1e-05
};
\addlegendentry{N-BCD}
\addplot [semithick, red!54.5098039215686!black, mark=asterisk, mark options={solid,fill opacity=0}]
table {%
10 8.26755544201841e-06
35.7894736842105 8.859963099244014e-06
61.5789473684211 8.87249461198241e-06
87.3684210526316 9.3071608579237e-06
113.157894736842 9.18370252386774e-06
138.947368421053 9.51710737916532e-06
164.736842105263 9.86746243592995e-06
190.526315789474 1e-05
216.315789473684 1e-05
242.105263157895 1e-05
267.894736842105 1e-05
293.684210526316 1e-05
319.473684210526 1e-05
345.263157894737 1e-05
371.052631578947 1e-05
396.842105263158 1e-05
422.631578947368 1e-05
448.421052631579 1e-05
474.210526315789 1e-05
500 1e-05
};
\addlegendentry{O-3}
\addplot [semithick, color0, mark=triangle, mark options={solid,rotate=180,fill opacity=0}]
table {%
10 8.0e-06
35.7894736842105 8.02749326747132e-06
61.5789473684211 8.03461289783804e-06
87.3684210526316 8.18521666571825e-06
113.157894736842 8.37601505823361e-06
138.947368421053 8.30722837229853e-06
164.736842105263 8.34786560898249e-06
190.526315789474 8.43734586784494e-06
216.315789473684 8.51536058490936e-06
242.105263157895 8.59035551617054e-06
267.894736842105 8.70e-6
293.684210526316 8.9e-06
319.473684210526 9.2e-06
345.263157894737 9.3e-06
371.052631578947 9.4e-06
396.842105263158 9.5e-06
422.631578947368 1e-05
448.421052631579 1e-05
474.210526315789 1e-05
500 1e-05
};
\addlegendentry{O-6}
\addplot [semithick, color1, mark=square, mark options={solid,fill opacity=0}]
table {%
10 5.5e-6
35.7894736842105 7.082782051243019e-06
61.5789473684211 7.396702384533e-06
87.3684210526316 7.47864159877952e-06
113.157894736842 8.4764885386639e-06
138.947368421053 8.57545713781539e-06
164.736842105263 8.59808592686239e-06
190.526315789474 8.5902031017426e-06
216.315789473684 8.59247186861084e-06
242.105263157895 8.59366816407762e-06
267.894736842105 8.7e-06
293.684210526316 8.9e-06
319.473684210526 9.2e-06
345.263157894737 1e-05
371.052631578947 1e-05
396.842105263158 1e-05
422.631578947368 1e-05
448.421052631579 1e-05
474.210526315789 1e-05
500 1e-05
};
\addlegendentry{O-9}
\end{axis}

\end{tikzpicture}
        \caption{Estimated outage $p_u$.}
        \label{fig:opt_outage_ge50}
    \end{subfigure}
    \caption{Average results obtained as a function of URLLC distance $d_u$, $d_e = 146.9$ m.}
    \label{fig:power-outage_ge50}
 \end{figure}
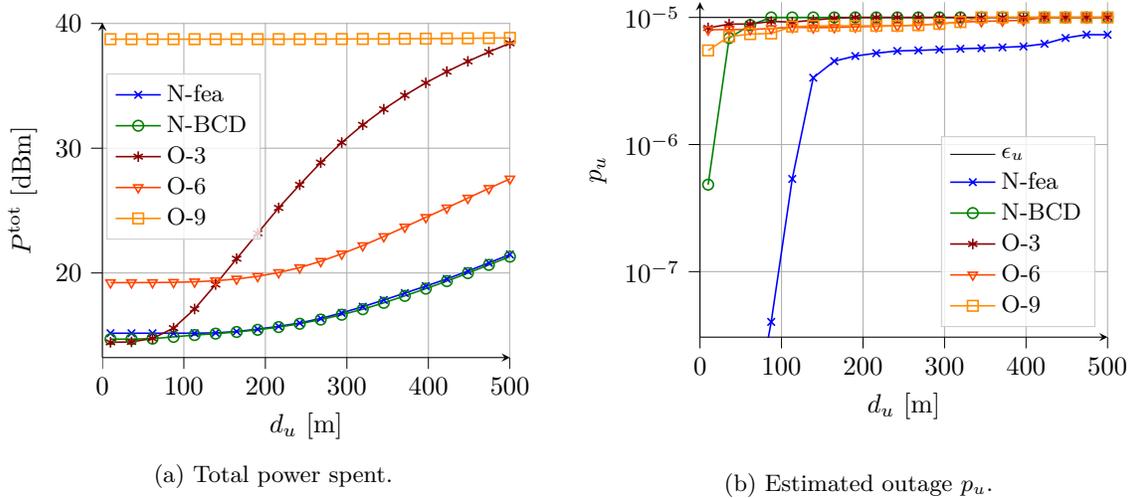

Fig.~\ref{fig:opt_power_ge50} shows the total power spent $P^\text{tot}$ as a function of $d_u$, when user $e$ is placed at $d_e = 146.9$ m (corresponding to $\Gamma_e = 50$ dB).
We can see that N-fea and N-BCD assure a lower power consumption than the OMA paradigms, as soon as $d_u \ge 75$ m. N-BCD performs better than OMA schemes also for $d_u \ge 65$ m. The performance gap between N-BCD and N-fea is negligible for every value of $d_u$ under analysis.
On the other hand, the OMA paradigm performs slightly better on short distances, i.e., $d_u < 65$ m. In particular, the best results in a high SNR regime are attained by O-3.
Moreover, Fig.~\ref{fig:opt_outage_ge50} presents the corresponding estimated $p_u$, evaluated with the allocated power coefficients $\mb{P}_u$ and $\mb{P}_e$, while the channel gains of $u$ are randomly generated. For the OMA schemes, the target outage is met, guaranteeing $p_u \le \epsilon_u$. The slight differences from the target outage obtained for short distances are due to the quantization of power available in the lookup table solution, set as 1 dB.
The outage probability of N-fea is always lower than $\epsilon_u$, proving that this scheme provides a feasible but not optimal solution. On the other hand, N-BCD reaches exactly the target outage probability until $d_u \ge 85$ m. For shorter distances, the power allocation is dominated by the SIC process limiting the minimum power needed, as we will see in the following. 

\begin{figure}
    \centering
    % This file was created with tikzplotlib v0.9.12.
\begin{tikzpicture}

\definecolor{color0}{rgb}{1,0.270588235294118,0}
\definecolor{color1}{rgb}{1,0.549019607843137,0}

\begin{axis}[
legend cell align={left},
legend style={
  fill opacity=0.8,
  draw opacity=1,
  text opacity=1,
  at={(0.97,0.03)},
  anchor=south east,
  draw=white!80!black
},
tick align=outside,
tick pos=left,
x grid style={white!69.0196078431373!black},
xlabel={\(\displaystyle d_u\) [m]},
xmajorgrids,
xmin=0, xmax=500,
xtick style={color=black},
y grid style={white!69.0196078431373!black},
ylabel={\(\displaystyle P_{u}^\mathrm{tot}\) [dBm]},
ymajorgrids,
ymin=-27.5510294630164, ymax=37.8598976196467,
ytick style={color=black}
]
\addplot [semithick, blue, mark=x, mark options={solid,fill opacity=0}]
table {%
10 8.16437144595252
35.7894736842105 8.16437206739808
61.5789473684211 8.16435930958481
87.3684210526316 8.16293510487428
113.157894736842 8.16971437419207
138.947368421053 8.30453818810225
164.736842105263 8.81776764144845
190.526315789474 9.59612553021966
216.315789473684 10.4144909092033
242.105263157895 11.2653270704237
267.894736842105 12.2658905062057
293.684210526316 13.322912325622
319.473684210526 14.3648325495484
345.263157894737 15.3618742126022
371.052631578947 16.299081456697
396.842105263158 17.1729661350396
422.631578947368 18.0099855622377
448.421052631579 18.8773486216606
474.210526315789 19.7414079979756
500 20.5780593857304
};
\addlegendentry{N-fea}
\addplot [semithick, green!50!black, mark=o, mark options={solid,fill opacity=0}]
table {%
10 4.98063720209382
35.7894736842105 4.98080634617838
61.5789473684211 5.32512536012742
87.3684210526316 6.43658786447372
113.157894736842 7.35784247631128
138.947368421053 7.9770081904725
164.736842105263 8.63832149011596
190.526315789474 9.38577698425366
216.315789473684 10.2100710774003
242.105263157895 11.0766502744233
267.894736842105 12.0256603329684
293.684210526316 12.9951473317445
319.473684210526 13.9584606875769
345.263157894737 14.9512300457033
371.052631578947 15.9283968632254
396.842105263158 16.8618919820094
422.631578947368 17.7575359667231
448.421052631579 18.6554701142672
474.210526315789 19.5301815960686
500 20.3656949391964
};
\addlegendentry{N-BCD}
\addplot [semithick, white!41.1764705882353!black, dashed]
table {%
826.11483599451 4.98072444921688
619.498731755915 4.98072444921688
464.558511632561 4.98072444921688
348.369737769362 4.98072444921688
261.24044905991 4.98072444921688
195.902699993437 4.98072444921688
146.906300317672 4.98072444921688
110.164183922677 4.98072444921688
82.611483599451 4.98072444921688
61.9498731755915 4.98072444921688
46.4558511632561 4.98072444921688
34.8369737769362 4.98072444921688
26.124044905991 4.98072444921688
19.5902699993437 4.98072444921688
14.6906300317672 4.98072444921688
};
\addlegendentry{SIC}
\addplot [semithick, black, dotted]
table {%
826.11483599451 7.28458246210814
619.498731755915 7.28458246210814
464.558511632561 7.28458246210814
348.369737769362 7.28458246210814
261.24044905991 7.28458246210814
195.902699993437 7.28458246210814
146.906300317672 7.28458246210814
110.164183922677 7.28458246210814
82.611483599451 7.28458246210814
61.9498731755915 7.28458246210814
46.4558511632561 7.28458246210814
34.8369737769362 7.28458246210814
26.124044905991 7.28458246210814
19.5902699993437 7.28458246210814
14.6906300317672 7.28458246210814
};
\addlegendentry{IL}
\addplot [semithick, red!54.5098039215686!black, mark=asterisk, mark size=2.5, mark options={solid,fill opacity=0}]
table {%
10 -24.5778055047136
35.7894736842105 -6.20527732742101
61.5789473684211 3.19471322588543
87.3684210526316 9.22108341487226
113.157894736842 13.7407383849723
138.947368421053 17.2340449702466
164.736842105263 20.1349676716927
190.526315789474 22.5777553551018
216.315789473684 24.7570814687083
242.105263157895 26.7068543563766
267.894736842105 28.7386787623885
293.684210526316 30.5928038816919
319.473684210526 32.1232561553263
345.263157894737 33.1158642746377
371.052631578947 33.7720138591464
396.842105263158 34.20477562611
422.631578947368 34.4827189207138
448.421052631579 34.6808994724534
474.210526315789 34.8136671581863
500 34.8866736613438
};
\addlegendentry{O-3}
\addplot [semithick, color0, mark=triangle, mark options={solid,rotate=180,fill opacity=0}]
table {%
10 -25.4637117061325
35.7894736842105 -18.9005696034033
61.5789473684211 -9.41193998448795
87.3684210526316 -3.33762677309339
113.157894736842 1.13646036215494
138.947368421053 4.66431414763419
164.736842105263 7.64490728860734
190.526315789474 10.1450145219049
216.315789473684 12.3520325441283
242.105263157895 14.2629678181684
267.894736842105 16.0696136814329
293.684210526316 17.6887518090728
319.473684210526 19.1374884037473
345.263157894737 20.506556344966
371.052631578947 21.7652885234294
396.842105263158 22.9076293455542
422.631578947368 23.9604038176035
448.421052631579 24.978489935929
474.210526315789 25.9434211366196
500 26.8464409941652
};
\addlegendentry{O-6}
\addplot [semithick, color1, mark=square, mark options={solid,fill opacity=0}]
table {%
10 -27.4464764162709
35.7894736842105 -20.4067817739271
61.5789473684211 -13.8392736642788
87.3684210526316 -7.68832506483348
113.157894736842 -3.16610351052543
138.947368421053 0.366087108152435
164.736842105263 3.32147929564205
190.526315789474 5.82945020609466
216.315789473684 8.09150072999331
242.105263157895 10.0385106303451
267.894736842105 11.8128301683309
293.684210526316 13.3807293249397
319.473684210526 14.7861544160268
345.263157894737 16.1416681349004
371.052631578947 17.4024774416099
396.842105263158 18.5532180083087
422.631578947368 19.6189977309026
448.421052631579 20.658417266568
474.210526315789 21.6472604788481
500 22.5735906693527
};
\addlegendentry{O-9}
\end{axis}

\end{tikzpicture}
    \caption{Average URLLC power spent as a function of $d_u$, $d_e = 146.9$.}
    \label{fig:urllc_powers_ge50}
    \vspace{-0.3cm}
\end{figure}
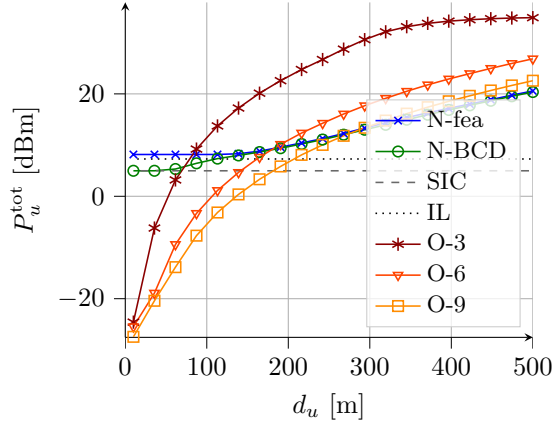

% can be explained remembering that only the eMBB user employs SIC, due to latency constraint. When the $u$ is closer to the BS than $e$, the power assuring the SIC requirement increases. In order to experimentally prove it, 
To explain the (slight) superiority of OMA to NOMA for small $d_u$ -or high $\Gamma_u$ regime-, we present the power needed for URLLC and eMBB requirements separately.
In Fig.~\ref{fig:urllc_powers_ge50}, we plot the power spent for the URLLC user $P_u^\text{tot} = \sum_{f\in\mc{F}_u} P_u(f)$.
In this figure, we show the power allocated using the various NOMA and OMA paradigms, as well as the power needed by the SIC process~\eqref{eq:power:sic}, namely SIC, and the power bound of the interference-limited scenario~\eqref{eq:power:u:IL}, namely IL. It is worth noting that these last two results depend on the eMBB power only.
The powers spent for the NOMA schemes are dominated by different effects for short distances. For N-fea, the interference-limited bound is the minimum power achievable. We remark that the power computed by~\eqref{eq:alg:fea} is obtained assuming that all the channels experience the maximum interference. Hence, the mutual information achievable using the power coefficient given by N-fea tends to the approximation~\eqref{eq:info:u:IL} on closer distances. %The IL power consumption is still marginally lower because the power coefficient $\mb{P}_u^\text{IL}$ are computed assuming lower interference power on average than the N-fea (see the derivation of~\eqref{eq:info:u:IL}).
On the other hand, N-BCD naturally exploits also the channels with no interference, enabling the possibility of reaching the power needed for the SIC process, which is the lower limit of the NOMA approach.
It is worth noting that even a small increment of URLLC power employed leads to a considerable reduction of the outage probability (see Fig.~\ref{fig:power-outage_ge50}) due to the non-linear relation between these quantities.
For the OMA schemes, the power spent decreases when resources reserved for the URLLC increase. For short distances, these schemes may consume less power w.r.t. NOMA  due to the lack of interference and lower limits. We remark that the frequency diversity gain provided for very short distances ($d_u < 35$ m) is negligible w.r.t. to the gain given by the mean SNR, resulting in comparable performances for the three OMA schemes.
Finally, in Table~\ref{tab:embb_power}, we present the average power spent for the eMBB. Here, the two NOMA schemes are not presented because the eMBB power allocation is the same for both N-fea and N-BCD. From the table, we can infer that the eMBB power plays the dominant role in the power allocation, explaining why O-9 total power is almost flat in Fig.~\ref{fig:opt_power_ge50}, even if it can benefit from the lowest URLLC power consumption. The same consideration can be made also for O-6. % However, the power spent by O-3 is comparable with the NOMA one. Hence, for $d_e \ge 146.9$, the best allocation scheme between NOMA and O-3 depends only on the URLLC power consumption, clarifying the results obtained in Fig.~\ref{fig:power-outage_ge50}.
However, O-3 and NOMA have comparable performance, and the best allocation scheme depends only on the URLLC power consumption for $d_e \ge 146.9$, clarifying the results obtained in Fig.~\ref{fig:power-outage_ge50}.

\begin{table}[b]
    \centering
    \caption{Average eMBB power spent in dBm.}
    \begin{tabular}{cr|rrrr}
\toprule
  $\Gamma_e$ [dB] & $d_e$ [m] &  NOMA &   O-3 &   O-6 &   O-9\\
\midrule
30 & 464.56 & 33.21 & 34.18 & 38.89 & 58.27 \\
40 & 261.2 & 23.36 & 24.32 & 29.09 & 48.54 \\
50 & 146.9 & 14.21 & 14.44 & 19.23 & 38.74 \\
60 & 82.6 & 5.84 &  5.84 &  9.05 & 28.44 \\
70 & 46.5 & -1.87 & -1.87 & -0.64 & 18.79 \\
80 & 26.1 & -9.67 & -9.67 & -9.67 &  8.65 \\
\bottomrule
\end{tabular}

    \label{tab:embb_power}
\end{table}

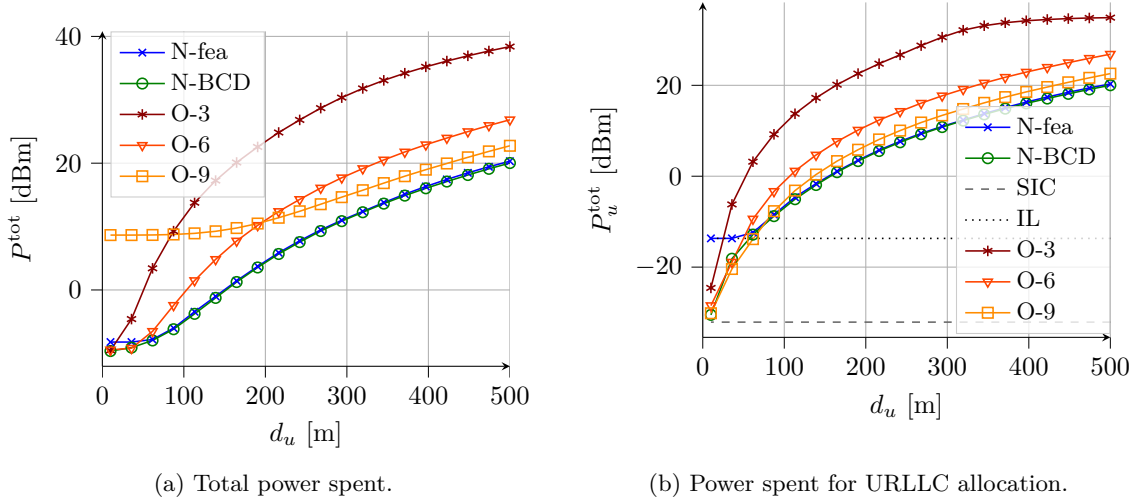
\begin{figure}[htb]
    \centering
    \begin{subfigure}[b]{0.45\textwidth}
        \centering
        % This file was created with tikzplotlib v0.9.12.
\begin{tikzpicture}

\definecolor{color0}{rgb}{1,0.270588235294118,0}
\definecolor{color1}{rgb}{1,0.549019607843137,0}

\begin{axis}[
legend cell align={left},
legend style={
  fill opacity=0.8,
  draw opacity=1,
  text opacity=1,
  at={(0.02,1)},
  anchor=north west,
  draw=white!80!black
},
tick align=outside,
tick pos=left,
x grid style={white!69.0196078431373!black},
xlabel={\(\displaystyle d_u\) [m]},
xmajorgrids,
xmin=0, xmax=500,
xtick style={color=black},
y grid style={white!69.0196078431373!black},
ylabel={\(\displaystyle P^\mathrm{tot}\) [dBm]},
ymajorgrids,
ymin=-12.0407406620165, ymax=40.796146736857,
ytick style={color=black}
]
\addplot [semithick, blue, mark=x, mark options={solid,fill opacity=0}]
table {%
10 -8.21972767916534
35.7894736842105 -8.21872753993189
61.5789473684211 -7.81274129830055
87.3684210526316 -6.01582792134604
113.157894736842 -3.43953038053804
138.947368421053 -1.00680256603869
164.736842105263 1.44049827492306
190.526315789474 3.72193999515157
216.315789473684 5.87182371601704
242.105263157895 7.75767275574649
267.894736842105 9.47685048342098
293.684210526316 11.0002126978084
319.473684210526 12.3862675918048
345.263157894737 13.784417190356
371.052631578947 15.1097564933252
396.842105263158 16.3258220632365
422.631578947368 17.4380810154932
448.421052631579 18.4735286085741
474.210526315789 19.4312308623211
500 20.3148826982378
};
\addlegendentry{N-fea}
\addplot [semithick, green!50!black, mark=o, mark options={solid,fill opacity=0}]
table {%
10 -9.63906396206772
35.7894736842105 -9.09833271619754
61.5789473684211 -7.96660625750494
87.3684210526316 -6.19572544985001
113.157894736842 -3.76993696031505
138.947368421053 -1.25453723397683
164.736842105263 1.25098082404199
190.526315789474 3.54528545191095
216.315789473684 5.67533637103849
242.105263157895 7.54813773676331
267.894736842105 9.29398624608532
293.684210526316 10.8544542049103
319.473684210526 12.2605471837262
345.263157894737 13.6190543206441
371.052631578947 14.8832680774511
396.842105263158 16.0371813479223
422.631578947368 17.1031424181705
448.421052631579 18.1337198942945
474.210526315789 19.1095855177537
500 20.0218846416991
};
\addlegendentry{N-BCD}
\addplot [semithick, red!54.5098039215686!black, mark=asterisk, mark options={solid,fill opacity=0}]
table {%
10 -9.53655702605354
35.7894736842105 -4.5920813222556
61.5789473684211 3.41335972477796
87.3684210526316 9.27499167824545
113.157894736842 13.7627732294564
138.947368421053 17.2540355005863
164.736842105263 20.1079475993002
190.526315789474 22.5619723835926
216.315789473684 24.8375829342392
242.105263157895 26.8247200948108
267.894736842105 28.6882485781872
293.684210526316 30.3453961238624
319.473684210526 31.7948757354952
345.263157894737 33.0670395424459
371.052631578947 34.1935294599047
396.842105263158 35.2008111634646
422.631578947368 36.1098299478588
448.421052631579 36.9369747796297
474.210526315789 37.6951188872076
500 38.3944700369082
};
\addlegendentry{O-3}
\addplot [semithick, color0, mark=triangle, mark options={solid,rotate=180,fill opacity=0}]
table {%
10 -9.4520928120247
35.7894736842105 -9.18448886749927
61.5789473684211 -6.53107548151286
87.3684210526316 -2.43001969910734
113.157894736842 1.48258948918448
138.947368421053 4.82134996216839
164.736842105263 7.72467376700635
190.526315789474 10.1900495488254
216.315789473684 12.3791809413323
242.105263157895 14.2804716571
267.894736842105 16.0811685852462
293.684210526316 17.6967139933575
319.473684210526 19.1431936112423
345.263157894737 20.5107196838051
371.052631578947 21.7684046917344
396.842105263158 22.9100249927128
422.631578947368 23.962283875346
448.421052631579 24.9799771799037
474.210526315789 25.9446121150939
500 26.8474084098068
};
\addlegendentry{O-6}
\addplot [semithick, color1, mark=square, mark options={solid,fill opacity=0}]
table {%
10 8.65687743531534
35.7894736842105 8.65692645610477
61.5789473684211 8.67594049818057
87.3684210526316 8.75125993203323
113.157894736842 8.92829875905359
138.947368421053 9.25249462536185
164.736842105263 9.76780785833403
190.526315789474 10.4760968041834
216.315789473684 11.390838964541
242.105263157895 12.4104577573744
267.894736842105 13.5239921522582
293.684210526316 14.6406594921265
319.473684210526 15.7326889218347
345.263157894737 16.8539778371203
371.052631578947 17.9460152988992
396.842105263158 18.976175337004
422.631578947368 19.9533566169806
448.421052631579 20.9237328747784
474.210526315789 21.8598486161755
500 22.7461436861171
};
\addlegendentry{O-9}
\end{axis}

\end{tikzpicture}
        \caption{Total power spent.}
        \label{fig:opt_power_ge80}
    \end{subfigure}
    ~
    \begin{subfigure}[b]{0.45\textwidth}
        \centering
        % This file was created with tikzplotlib v0.9.12.
\begin{tikzpicture}

\definecolor{color0}{rgb}{1,0.270588235294118,0}
\definecolor{color1}{rgb}{1,0.549019607843137,0}

\begin{axis}[
legend cell align={left},
legend style={
  fill opacity=0.8,
  draw opacity=1,
  text opacity=1,
  at={(1,0.01)},
  anchor=south east,
  draw=white!80!black
},
tick align=outside,
tick pos=left,
x grid style={white!69.0196078431373!black},
xlabel={\(\displaystyle d_u\) [m]},
xmajorgrids,
xmin=0, xmax=500,
xtick style={color=black},
y grid style={white!69.0196078431373!black},
ylabel={\(\displaystyle P_{u}^\mathrm{tot}\) [dBm]},
ymajorgrids,
ymin=-35.4517025921071, ymax=38.2361201496034,
ytick style={color=black}
]
\addplot [semithick, blue, mark=x, mark options={solid,fill opacity=0}]
table {%
10 -13.6759989259186
35.7894736842105 -13.6724868654733
61.5789473684211 -12.388489238304
87.3684210526316 -8.46188041487788
113.157894736842 -4.61966009073049
138.947368421053 -1.64113430777384
164.736842105263 1.09085881382909
190.526315789474 3.51859228280263
216.315789473684 5.74900926939638
242.105263157895 7.67851560627352
267.894736842105 9.42372854460981
293.684210526316 10.9628747999222
319.473684210526 12.3591635965994
345.263157894737 13.7647906018595
371.052631578947 15.0953003470656
396.842105263158 16.3149008895511
422.631578947368 17.4296297845913
448.421052631579 18.4668715017185
474.210526315789 19.4258919915838
500 20.3105272278884
};
\addlegendentry{N-fea}
\addplot [semithick, green!50!black, mark=o, mark options={solid,fill opacity=0}]
table {%
10 -30.5052298288847
35.7894736842105 -18.1535641303458
61.5789473684211 -12.8452423379621
87.3684210526316 -8.78296690747432
113.157894736842 -5.05860335597375
138.947368421053 -1.92914805102084
164.736842105263 0.885070589228212
190.526315789474 3.33328642688967
216.315789473684 5.54675286848188
242.105263157895 7.46503017249972
267.894736842105 9.23856515659119
293.684210526316 10.8158362186172
319.473684210526 12.232644547151
345.263157894737 13.5986642256138
371.052631578947 14.8680366671846
396.842105263158 16.0255086581044
422.631578947368 17.0940129026758
448.421052631579 18.1265205406964
474.210526315789 19.1038359596069
500 20.0172250263004
};
\addlegendentry{N-BCD}
\addplot [semithick, white!41.1764705882353!black, dashed]
table {%
500 -32.1022561038476
10 -32.1022561038476
};
\addlegendentry{SIC}
\addplot [semithick, black, dotted]
table {%
500 -13.6819989259186
10 -13.6819989259186
};
\addlegendentry{IL}
\addplot [semithick, red!54.5098039215686!black, mark=asterisk, mark options={solid,fill opacity=0}]
table {%
10 -24.5778055047136
35.7894736842105 -6.20527732742101
61.5789473684211 3.19471322588543
87.3684210526316 9.22108341487226
113.157894736842 13.7407383849723
138.947368421053 17.2340449702466
164.736842105263 20.1349676716927
190.526315789474 22.5777553551018
216.315789473684 24.7570814687083
242.105263157895 26.7068543563766
267.894736842105 28.7386787623885
293.684210526316 30.5928038816919
319.473684210526 32.1232561553263
345.263157894737 33.1158642746377
371.052631578947 33.7720138591464
396.842105263158 34.20477562611
422.631578947368 34.4827189207138
448.421052631579 34.6808994724534
474.210526315789 34.8136671581863
500 34.8866736613438
};
\addlegendentry{O-3}
\addplot [semithick, color0, mark=triangle, mark options={solid,rotate=180,fill opacity=0}]
table {%
10 -28.4637117061325
35.7894736842105 -18.9005696034033
61.5789473684211 -9.41193998448795
87.3684210526316 -3.33762677309339
113.157894736842 1.13646036215494
138.947368421053 4.66431414763419
164.736842105263 7.64490728860734
190.526315789474 10.1450145219049
216.315789473684 12.3520325441283
242.105263157895 14.2629678181684
267.894736842105 16.0696136814329
293.684210526316 17.6887518090728
319.473684210526 19.1374884037473
345.263157894737 20.506556344966
371.052631578947 21.7652885234294
396.842105263158 22.9076293455542
422.631578947368 23.9604038176035
448.421052631579 24.978489935929
474.210526315789 25.9434211366196
500 26.8464409941652
};
\addlegendentry{O-6}
\addplot [semithick, color1, mark=square, mark options={solid,fill opacity=0}]
table {%
10 -30.14464764162709
35.7894736842105 -20.4067817739271
61.5789473684211 -13.8392736642788
87.3684210526316 -7.68832506483348
113.157894736842 -3.16610351052543
138.947368421053 0.366087108152435
164.736842105263 3.32147929564205
190.526315789474 5.82945020609466
216.315789473684 8.09150072999331
242.105263157895 10.0385106303451
267.894736842105 11.8128301683309
293.684210526316 13.3807293249397
319.473684210526 14.7861544160268
345.263157894737 16.1416681349004
371.052631578947 17.4024774416099
396.842105263158 18.5532180083087
422.631578947368 19.6189977309026
448.421052631579 20.658417266568
474.210526315789 21.6472604788481
500 22.5735906693527
};
\addlegendentry{O-9}
\end{axis}

\end{tikzpicture}
        \caption{Power spent for URLLC allocation.}
        \label{fig:opt_urllc_power_ge50}
     \end{subfigure}
     \caption{Average results obtained as a function of URLLC distance $d_u$, $d_e = 26.1$ m.}
     \label{fig:powers_ge80}
\end{figure}

In Fig.~\ref{fig:powers_ge80}, we present the power results in the extreme case of $d_e = 26.1$ m. In this case, the allocated eMBB power is low, and it is the same for NOMA, O-3, and O-6, as shown in Table~\ref{tab:embb_power}. With this power allocation, the N-BCD approach is never limited by the SIC, and its performances are the best for every value of $d_u$. Note that in this extreme case, the best performances between the OMA schemes are attained by O-6, which consumes the same eMBB power but benefits from the frequency diversity gain for URLLC allocation.

% \begin{figure}[htb]
%     \centering
%     \begin{subfigure}[c]{0.45\textwidth}
%         \centering
%         \input{journalplots/optimization_power_Mu1_ru0.1_re4.0_snre40}
%         \caption{Total power spent.}
%         \label{fig:opt_power_ge40}
%     \end{subfigure}
%     ~
%     \begin{subfigure}[c]{0.45\textwidth}
%         \centering
%         \input{journalplots/optimization_u_powers_Mu1_ru0.1_re4.0_snre40}
%         \caption{Power spent for URLLC allocation.}
%         \label{fig:opt_urllc_power_ge40}
%      \end{subfigure}
%      \caption{Average results obtained as a function of URLLC distance $d_u$, $d_e = 261.2$ m.}
%      \label{fig:powers_ge40}
% \end{figure}

% The last case shown is depicted in Fig.~\ref{fig:powers_ge40}, where total and URLLC power consumption are shown when $d_e = 261.2$. In this case, O-9 is not shown because of the enormous power spent for the eMBB allocation (see Table~\ref{tab:embb_power}). Once again, NOMA schemes outperform OMA schemes.
% It is interesting to note that SIC bound is greater than the IL bound; the power allocation is dominated by the SIC as soon the URLLC is the user closest to the BS. For this scenario, the simple solution~\eqref{eq:power:sic} can be employed, guaranteeing the requirements of both users.

\begin{figure}
    \centering
    \begin{subfigure}[b]{0.45\textwidth}
        \centering
        % This file was created with tikzplotlib v0.9.12.
\begin{tikzpicture}

\definecolor{color0}{rgb}{1,0.270588235294118,0}
\definecolor{color1}{rgb}{1,0.549019607843137,0}

\begin{axis}[
legend cell align={left},
legend style={
  fill opacity=0.8,
  draw opacity=1,
  text opacity=1,
  at={(0.03,0.97)},
  anchor=north west,
  draw=white!80!black
},
tick align=outside,
tick pos=left,
x grid style={white!69.0196078431373!black},
xlabel={\(\displaystyle d_e\) [m]},
xmajorgrids,
xmin=0, xmax=500,
xtick style={color=black},
y grid style={white!69.0196078431373!black},
ylabel={\(\displaystyle P^\mathrm{tot}\) [dBm]},
ymajorgrids,
ymin=-3.32551389508208, ymax=59.2315729399766,
ytick style={color=black}
]
\addplot [semithick, blue, mark=x,  mark options={solid,fill opacity=0}]
table {%
% 10 -0.482009948033958
35.7894736842105 0.922998929297963
61.5789473684211 4.70544094786111
87.3684210526316 8.1943103506262
113.157894736842 11.1425358516609
138.947368421053 14.2418107279149
164.736842105263 17.2243014764095
190.526315789474 19.6228742988509
216.315789473684 21.5618421974277
242.105263157895 23.1715573480657
267.894736842105 24.5414539769746
293.684210526316 25.7311111040454
319.473684210526 26.7811817076895
345.263157894737 27.7203134969715
371.052631578947 28.5693264408351
396.842105263158 29.3437715663386
422.631578947368 30.0555441602076
448.421052631579 30.713932723639
474.210526315789 31.3263213201121
500 31.8986725823469
};
\addlegendentry{N-fea}
\addplot [semithick, green!50!black, mark=o,  mark options={solid,fill opacity=0}]
table {%
% 10 -0.2765728591648
35.7894736842105 0.618248251534864
61.5789473684211 4.53767648622902
87.3684210526316 8.13485469728555
113.157894736842 11.1118413031869
138.947368421053 14.1985520302479
164.736842105263 17.1589537449158
190.526315789474 19.5431788433703
216.315789473684 21.4730385955263
242.105263157895 23.0766474401226
267.894736842105 24.4422218998936
293.684210526316 25.6286797196257
319.473684210526 26.6762956751182
345.263157894737 27.6134889433206
371.052631578947 28.4609344625232
396.842105263158 29.2340872704863
422.631578947368 29.9447768093078
448.421052631579 30.6022450127764
474.210526315789 31.2138421344241
500 31.7855056949937
};
\addlegendentry{N-BCD}
\addplot [semithick, red!54.5098039215686!black, mark=asterisk,  mark options={solid,fill opacity=0}]
table {%
% 10 18.2133578797863
35.7894736842105 18.2116510122633
61.5789473684211 18.2991561812261
87.3684210526316 18.4817876189684
113.157894736842 18.7640934264239
138.947368421053 19.4062384781266
164.736842105263 20.5438598789971
190.526315789474 21.8643193925342
216.315789473684 23.1815226088429
242.105263157895 24.4189189388626
267.894736842105 25.5553277150271
293.684210526316 26.5920691422009
319.473684210526 27.5381315193614
345.263157894737 28.4042055158916
371.052631578947 29.2004841400659
396.842105263158 29.9359769484089
422.631578947368 30.618408205386
448.421052631579 31.2543226056535
474.210526315789 31.8492488903557
500 32.4078642590026
};
\addlegendentry{O-3}
\addplot [semithick, color0, mark=triangle,  mark options={solid,rotate=180,fill opacity=0}]
table {%
% 10 6.53409590967032
35.7894736842105 5.91888602328006
61.5789473684211 7.99461802399403
87.3684210526316 11.3534625540658
113.157894736842 14.623586923877
138.947368421053 18.2715153139949
164.736842105263 21.6719580496172
190.526315789474 24.2884709503164
216.315789473684 26.3509094980524
242.105263157895 28.037762222476
267.894736842105 29.4596759340913
293.684210526316 30.6864924182727
319.473684210526 31.7643060891472
345.263157894737 32.7248781337768
371.052631578947 33.5909237833405
396.842105263158 34.3792076127021
422.631578947368 35.1024351097029
448.421052631579 35.7704549405065
474.210526315789 36.3910501047753
500 36.9704746485756
};
\addlegendentry{O-6}
\addplot [semithick, color1, mark=square,  mark options={solid,fill opacity=0}]
table {%
% 10 17.9729136021556
35.7894736842105 17.9729136021556
61.5789473684211 23.5565046978324
87.3684210526316 29.4233327349689
113.157894736842 33.5688591777302
138.947368421053 37.5547588202472
164.736842105263 41.0536827128097
190.526315789474 43.6983272051308
216.315789473684 45.7699798969681
242.105263157895 47.4598917297645
267.894736842105 48.8825637387634
293.684210526316 50.1092101993979
319.473684210526 51.1864743982358
345.263157894737 52.1463546634595
371.052631578947 53.0116733909708
396.842105263158 53.7992441626649
422.631578947368 54.5217941261696
448.421052631579 55.1891803435051
474.210526315789 55.809187703227
500 56.3880689929285
};
\addlegendentry{O-9}
\end{axis}

\end{tikzpicture}
        \caption{Total power spent.}
        \label{fig:opt_power_gu50}
    \end{subfigure}
    ~
    \begin{subfigure}[b]{0.45\textwidth}
        \centering
        % This file was created with tikzplotlib v0.9.12.
\begin{tikzpicture}

\definecolor{color0}{rgb}{1,0.270588235294118,0}
\definecolor{color1}{rgb}{1,0.549019607843137,0}

\begin{axis}[
legend cell align={left},
legend style={
  fill opacity=0.8,
  draw opacity=1,
  text opacity=1,
  at={(0.97,0.03)},
  anchor=south east,
  draw=white!80!black
},
tick align=outside,
tick pos=left,
x grid style={white!69.0196078431373!black},
xlabel={\(\displaystyle d_e\) [m]},
xmajorgrids,
xmin=0, xmax=500,
xtick style={color=black},
y grid style={white!69.0196078431373!black},
ylabel={\(\displaystyle P_{u}^\mathrm{tot}\) [dBm]},
ymajorgrids,
ymin=-29.2941418488635, ymax=26.7725991813794,
ytick style={color=black}
]
\addplot [semithick, blue, mark=x,  mark options={solid,fill opacity=0}]
table {%
% 10 -1.92497137826728
35.7894736842105 -0.12272695556298
61.5789473684211 1.38871438553063
87.3684210526316 3.11046830752683
113.157894736842 5.15125030440302
138.947368421053 7.60033636448676
164.736842105263 10.1561691546778
190.526315789474 12.3328183543729
216.315789473684 14.1480958518888
242.105263157895 15.6823656130457
267.894736842105 17.0026716783306
293.684210526316 18.1577683851896
319.473684210526 19.1826230234462
345.263157894737 20.1026735209652
371.052631578947 20.9368153486435
396.842105263158 21.6993864146511
422.631578947368 22.4014845832489
448.421052631579 23.0518552905562
474.210526315789 23.6575014109351
500 24.224110952732
};
\addlegendentry{N-fea}
\addplot [semithick, green!50!black, mark=o,  mark options={solid,fill opacity=0}]
table {%
% 10 -1.64116652082613
35.7894736842105 -0.514288303305455
61.5789473684211 1.0203306927916
87.3684210526316 2.91579551905174
113.157894736842 5.02799541688013
138.947368421053 7.39701350021641
164.736842105263 9.81273867932403
190.526315789474 11.8876451517701
216.315789473684 13.6342766914768
242.105263157895 15.1209806142585
267.894736842105 16.4069788034631
293.684210526316 17.536392749196
319.473684210526 18.5414041100367
345.263157894737 19.4457137199962
371.052631578947 20.2670913098615
396.842105263158 21.0191188188238
422.631578947368 21.7123701530012
448.421052631579 22.3552173804124
474.210526315789 22.9543907707433
500 23.5153750029077
};
\addlegendentry{N-BCD}
\addplot [semithick, white!41.1764705882353!black, dashed]
table {%
% 10 -19.3672371959622
35.7894736842105 -26.7456536202161
61.5789473684211 -14.4734827822188
87.3684210526316 -6.53275538743987
113.157894736842 -1.36576207714548
138.947368421053 3.57481618588495
164.736842105263 7.6143392071526
190.526315789474 10.5002119591562
216.315789473684 12.6996866012527
242.105263157895 14.4668962140481
267.894736842105 15.9407858943918
293.684210526316 17.2036503315793
319.473684210526 18.3077880244289
345.263157894737 19.2883560302169
371.052631578947 20.1700688759443
396.842105263158 20.9709382438309
422.631578947368 21.7044848191614
448.421052631579 22.3811111225468
474.210526315789 23.0089888276433
500 23.5946523400084
};
\addlegendentry{SIC}
\addplot [semithick, black, dotted]
table {%
% 10 -14.6279640970088
35.7894736842105 -9.21408635406441
61.5789473684211 -2.94394031591848
87.3684210526316 1.00461569352293
113.157894736842 3.93226312901758
138.947368421053 6.51625989899819
164.736842105263 8.87572129270786
190.526315789474 10.8569383007898
216.315789473684 12.5212250701889
242.105263157895 13.9418942288941
267.894736842105 15.1754746020553
293.684210526316 16.2628656990267
319.473684210526 17.2336847415863
345.263157894737 18.1097476402983
371.052631578947 18.9074577961382
396.842105263158 19.6394004385844
422.631578947368 20.3154127329626
448.421052631579 20.9433138170208
474.210526315789 21.5294125791889
500 22.0788678389638
};
\addlegendentry{IL}
\addplot [semithick, red!54.5098039215686!black, mark=asterisk,  mark options={solid,fill opacity=0}]
table {%
% 10 18.19463597062
35.7894736842105 18.19463597062
61.5789473684211 18.19463597062
87.3684210526316 18.19463597062
113.157894736842 18.19463597062
138.947368421053 18.19463597062
164.736842105263 18.19463597062
190.526315789474 18.19463597062
216.315789473684 18.19463597062
242.105263157895 18.19463597062
267.894736842105 18.19463597062
293.684210526316 18.19463597062
319.473684210526 18.19463597062
345.263157894737 18.19463597062
371.052631578947 18.19463597062
396.842105263158 18.19463597062
422.631578947368 18.19463597062
448.421052631579 18.19463597062
474.210526315789 18.19463597062
500 18.19463597062
};
\addlegendentry{O-3}
\addplot [semithick, color0, mark=triangle,  mark options={solid,rotate=180,fill opacity=0}]
table {%
% 10 5.64938037170428
35.7894736842105 5.64938037170428
61.5789473684211 5.64938037170428
87.3684210526316 5.64938037170428
113.157894736842 5.64938037170428
138.947368421053 5.64938037170428
164.736842105263 5.64938037170428
190.526315789474 5.64938037170428
216.315789473684 5.64938037170428
242.105263157895 5.64938037170428
267.894736842105 5.64938037170428
293.684210526316 5.64938037170428
319.473684210526 5.64938037170428
345.263157894737 5.64938037170428
371.052631578947 5.64938037170428
396.842105263158 5.64938037170428
422.631578947368 5.64938037170428
448.421052631579 5.64938037170427
474.210526315789 5.64938037170427
500 5.64938037170427
};
\addlegendentry{O-6}
\addplot [semithick, color1, mark=square,  mark options={solid,fill opacity=0}]
table {%
% 10 1.344226896195
35.7894736842105 1.344226896195
61.5789473684211 1.344226896195
87.3684210526316 1.344226896195
113.157894736842 1.344226896195
138.947368421053 1.344226896195
164.736842105263 1.344226896195
190.526315789474 1.344226896195
216.315789473684 1.344226896195
242.105263157895 1.344226896195
267.894736842105 1.344226896195
293.684210526316 1.344226896195
319.473684210526 1.344226896195
345.263157894737 1.34422689619499
371.052631578947 1.34422689619499
396.842105263158 1.34422689619499
422.631578947368 1.34422689619499
448.421052631579 1.34422689619499
474.210526315789 1.34422689619499
500 1.34422689619499
};
\addlegendentry{O-9}
\end{axis}

\end{tikzpicture}
        \caption{Power spent for URLLC allocation.}
        \label{fig:opt_urllc_power_gu50}
     \end{subfigure}
     \caption{Average results obtained as a function of eMBB distance $d_e$, $d_u = 146.9$ m.}
     \label{fig:powers_gu50}
\end{figure}
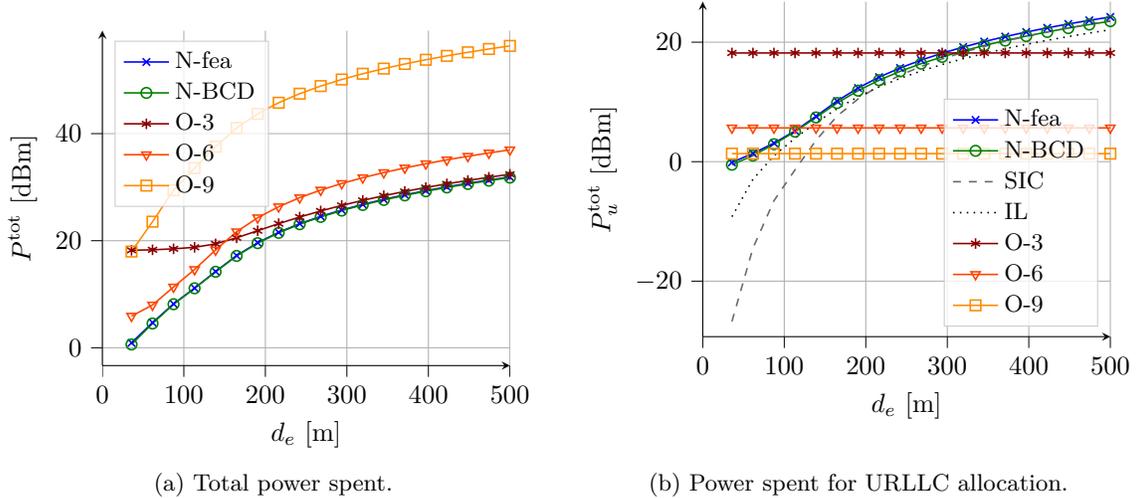

% \begin{figure}
%     \centering
%     \subfloat[Total power spent. \label{fig:opt_power_gu50}]{\input{journalplots/optimization_power_Mu1_ru0.1_re4.0_snru50}}
%     \hfill
%     \subfloat[Power spent for URLLC allocation.
%         \label{fig:opt_urllc_power_gu50}]{\input{journalplots/optimization_u_powers_Mu1_ru0.1_re4.0_snru50}}
%      \caption{Average results obtained as a function of eMBB distance $d_e$, $d_u = 146.9$ m.}
%      \label{fig:powers_gu50}
% \end{figure}
% Finally, we show i
In Fig.~\ref{fig:powers_gu50}, we show the power spent as a function of $d_e$, fixing $d_u = 146.9$ m. 
In detail, Fig.~\ref{fig:opt_power_gu50} shows the overall power spent; also in this case, the NOMA schemes attain the best performance. Fig.~\ref{fig:opt_power_gu50} shows the URLLC power spent to meet the reliability requirement. In this case, the OMA power does not change for the different values of $d_e$, having fixed $d_u$. On the contrary, NOMA URLLC power is still influenced by the interference given by $\mb{P}_e$, and thus $P_u^\text{tot}$ increases when $d_e$ increases, accordingly. 

% We remark that the feasible solution provide by Algorithm~\ref{alg:N-fea} is enough to obtain close-to-the-optimal performances for every value of $d_e$ and $d_u$. Therefore, a practical implementation of the spectrum slicing scheduler can make use of the pre-populated lookup table~\eqref{eq:montecarlopu}, and still lead to a negligible increase of power consumption w.r.t. the optimal solution.

\begin{figure}
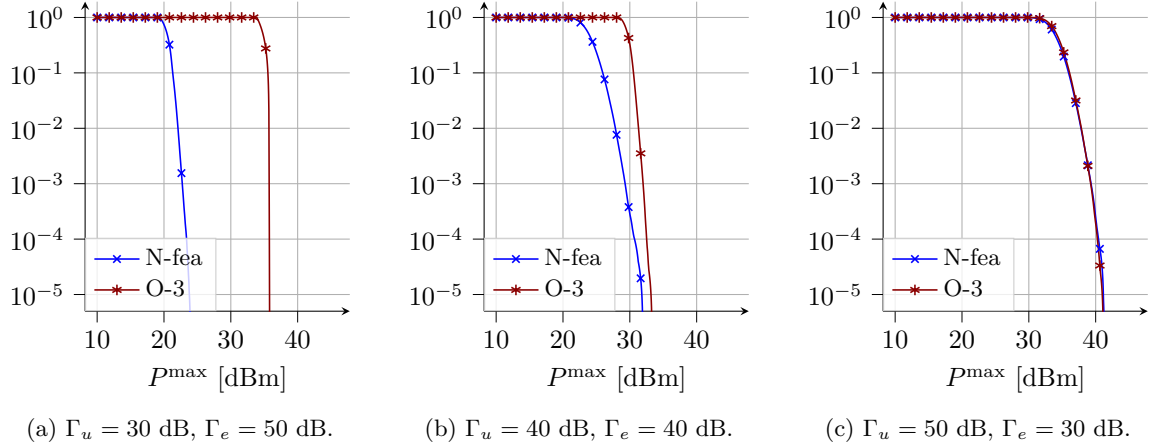

    \centering
    \begin{subfigure}[t]{0.3\textwidth}
        \centering
        \input{journalplots/PowerOutage_snru30_snre50}
        \caption{$\Gamma_u = 30$ dB, $\Gamma_e = 50$ dB.}
        \label{fig:PowerOutage3050}
    \end{subfigure}
    ~
    \begin{subfigure}[t]{0.3\textwidth}
        \centering
        \input{journalplots/PowerOutage_snru40_snre40}
        \caption{$\Gamma_u = 40$ dB, $\Gamma_e = 40$ dB.}
        \label{fig:PowerOutage4040}
     \end{subfigure}
     ~
     \begin{subfigure}[t]{0.3\textwidth}
        \centering
        \input{journalplots/PowerOutage_snru50_snre30}
        \caption{$\Gamma_u = 50$ dB, $\Gamma_e = 30$ dB.}
        \label{fig:PowerOutage5030}
     \end{subfigure}
     \caption{Probability of outage due to the power budget, for different value of $\Gamma_e$ and $\Gamma_u$.}
     \label{fig:PowerOutage}
\end{figure}

\FS{Finally, Fig.~\ref{fig:PowerOutage} shows the impact of considering the BS power budget $P^\text{max}$ for N-fea and O-3. We plot the probability that the total power allocated is higher than the power budget, i.e, $\Pr\{P^{\text{tot}} > P^{\text{max}}\}$, varying the latter, and with different combinations of $\Gamma_u$ and $\Gamma_e$. The NOMA approach can provide better performance if $\Gamma_u \le \Gamma_e$, i.e., when the eMBB is the \emph{strong} user, accordingly with the imposition of employing the SIC process at such user.}

\section{Conclusions \mm{and future directions}}
\label{sec:conclusions}
This paper studied the power and resource allocation for the downlink communication spectrum slicing of eMBB and URLLC traffic. We focused on orthogonal and non-orthogonal multiple access schemes employing the parallel channel model.
Due to the nature of the communication traffic types, we assumed that the CSI is only statistical for the URLLC, while the eMBB channel relies on instantaneous information.
We proposed a feasible and a BCD algorithm able to solve the spectrum slicing problem by minimizing the power spent, assuring at the same time the requirements of both kinds of traffic. We also compared the impact of using the aforementioned multiple access scheme for a 5G-like resource grid available.
Numerical results showed that the NOMA paradigm attains the best performance in almost all cases. The only exception is for a very close URLLC user, where the frequency diversity gain is negligible. In that case, OMA could attain the best performance depending on the position of the eMBB user. However, the performance gain w.r.t. NOMA is still negligible, proving that NOMA is a promising technology for spectrum slicing.
We also proved experimentally that the lookup table algorithm provides a close-to-the-optimal power allocation, becoming a possible candidate for a practical resource allocation approach.

This work is a first step to understanding the optimal design of a beyond 5G network enabling spectrum slicing on downlink communications. \FS{Based on our promising results and insights, our future work will investigate spectrum slicing in a more practical scenario with several users for the two types of traffic, considering both single-antenna and multiple-antenna technologies.}

\appendix
\section{Proof of Proposition~\ref{theo:non-increasing}}
\label{proof:non-increasing}
The $\log$ function is monotonic non-decreasing w.r.t. each $P_u(f)$. % Hence, reducing the power coefficient on a single mRB keeping the other power terms unchanged will reduce the mutual information $I_u$. 
% More formally, 
If we consider the vector $\mb{P}_u'$ obtained decreasing $P_u(f)$, $f\in\mc{F}_u$, by $\delta > 0$, while the other $P_u(g)$, $\forall g\in\mc{F}_u \setminus \{f\}$, and $P_e(j)$, $\forall j\in\mc{F}_u$, are kept the same, the mutual information results
\[
    F_u I_u(\mb{P}_u') = \log_2\left( 1 + \frac{\gamma_u(f) (P_u(f) - \delta)}{1 + \gamma_u(f) P_e(f)} \right) + \hspace{-2mm} \sum_{g\in\mc{F}_u\setminus \{f\}} \log_2\left( 1 + \frac{\gamma_u(f)P_u(f)}{1 + \gamma_u(f) P_e(f)} \right) \le F_u I_u(\mb{P}_u).
\]
Therefore, the outage probability may only result larger (or equal), i.e., $p_u(\mb{P}_u') = \Pr\{I_u(\mb{P}_u') \le r_u\} \ge p_u(\mb{P}_u) = \Pr\{I_u(\mb{P}_u) \le r_u\}$.
% \[
%     p_u(\mb{P}_u') = \Pr\{I_u(\mb{P}_u') \le r_u\} \ge p_u(\mb{P}_u) = \Pr\{I_u(\mb{P}_u) \le r_u\}.
% \]
The same relation is in fact extended for each $\mb{0} \preceq \mb{P}_u' \preceq \mb{P}_u$, which completes the proof. \qed

\section{Proof of Proposition~\ref{theo:oma-optimality}}
\label{proof:oma-optimality}
For the OMA allocation, the power coefficient of $e$ user is zero for every resources given to $u$, i.e. $P_e(f) = 0$, $\forall f\in\mc{F}_u$. %Hence, every channel experiences the same channel condition. 
In these conditions of i.i.d. parallel channels without interference, the minimum outage probability is reached when the same power coefficient is allocated to each channel~\cite{Tse2005}.
As a matter of fact, Algorithm~\ref{alg:N-fea} outputs the power coefficient $P_u$ given by~\eqref{eq:fea:minP} for each mRB, due to the lack of SIC constraint. % for the OMA case. 
The resulting vector $\mb{P}_u^* = [P_u, \dots, P_u]\T$ is at least a feasible solution of problem~\eqref{op:oma}, due to Proposition~\ref{theo:feasible}.
The resulting mutual information is $I_u(\mb{P}_u^*, \mb{P}_e) = I_u(P_u, 0) = \frac{1}{F_u} \sum_{f\in\mc{F}_u} \log_2\left(1 + \gamma_u(f) P_u \right)$.
% \[
% I_u(\mb{P}_u^*, \mb{P}_e) = I_u(P_u, 0) = \frac{1}{F_u} \sum_{f\in\mc{F}_u} \log_2\left(1 + \gamma_u(f) P_u \right).
% \]
Following Proposition~\ref{theo:non-increasing}, the outage probability $p_u(\mb{P}_u^*)$ is monotonically non-increasing in $P_u$. %Moreover, the objective function of~\eqref{eq:op:u} is monotonically non-decreasing
Having obtained $P_u$ from~\eqref{eq:fea:minP}, there is no other power coefficient $P_u' \le P_u$ such as $\hat{p}_u(P_u', 0, \Gamma_u, \mc{F}_u, r_u) \le \epsilon_u$. 
If the table has been populated with enough different values of $P_u$, we will find that $p_u(P_u, 0, \Gamma_u, \mc{F}_u, r_u) = \epsilon_u$, and the optimal solution is reached. \qed

\section{Proof of Proposition~\ref{theo:convergence}}
\label{proof:convergence}
The objective function of problem~\eqref{eq:op:u} is a non-decreasing function with gradient equal to $\mb{1}$.
% \begin{equation}
%     \nabla_{\mb{P}_u} \sum_{f\in\mc{F}_u} P_u(f) = \mb{1}.
% \end{equation}
Hence, no saddle point can be found and the optimal solution lay on the border of the feasible set~\cite{Boyd}.
% Moreover, the optimal point is the one nearest to the origin, considering that the minimization direction is given by the vector $-\mb{1}$.
The power updating rule~\eqref{eq:pupdate} guarantees that the magnitude of vector $\mb{P}_u$ decreases at every iteration, following the minimization direction.
Let us suppose that the optimal solution is $\mb{P}_u^\text{o}$, and the Algorithm has stopped on point $\mb{P}_u^*$ when the threshold is set as $\tau'$.
If $\mb{P}_u^\text{o} = \mb{P}_u^*$, the proof is completed.
Let us now assume that the vector $\mb{P}_u^*$ has the same optimal coefficients except for one dimension $f$, i.e. $P_u^\text{o}(f) < P_u^*(f)$, and $P_u^\text{o}(g) = P_u^*(g)$ $\forall g \in\mc{F}_u \setminus \{f\}$.
This means that that exists a feasible point on dimension $f\in\mc{F}_u$ that will provide a smaller objective function. Hence, considering $\tau'' < \tau'$, using both~\eqref{eq:mupdate} and~\eqref{eq:pupdate}, it is possible to reduce $P_u^*(f)$ toward $P_u^\text{o}(f)$. Reducing $\tau''$, there exists a value of the threshold $0< \tau^* \le \tau''$ such as updating rules~\eqref{eq:mupdate} and~\eqref{eq:pupdate} will provide exactly $P_u^\text{o}(f)$.
The same results can be extended for the general case where $\mb{P}_u^\text{o}$ and $\mb{P}_u^*$ differ by more than one element, completing the proof. \qed

%===========================
\bibliographystyle{IEEEtranNoUrl}
\bibliography{IEEEabrv, slicing.bib}
\end{document}